\documentclass[fleqn,10pt]{wlscirep}

\usepackage{cite}

\usepackage{cite}

\graphicspath{{./}}
\DeclareGraphicsExtensions{.pdf}

\usepackage{mathtools}
\usepackage{color}
\usepackage{amssymb}
\usepackage{amsthm}
\usepackage{accents}
\usepackage{algorithm}
\DeclareMathAlphabet{\mathcal}{OMS}{cmsy}{m}{n}

\usepackage{soul}

\usepackage{caption}
\usepackage[labelformat=simple]{subcaption}

\DeclareMathAlphabet{\mathpzc}{OT1}{pzc}{m}{it}

\DeclarePairedDelimiter\abs{\lvert}{\rvert}%
\DeclarePairedDelimiter\norm{\lVert}{\rVert}%

\makeatletter
\let\oldabs\abs
\def\abs{\@ifstar{\oldabs}{\oldabs*}}
\let\oldnorm\norm
\def\norm{\@ifstar{\oldnorm}{\oldnorm*}}

\newcommand\by[1]{\mathbf{r}_{#1}}

\newcommand\bn[1]{\mathbf{n}_{#1}}

\newcommand\bx[1]{\mathbf{u}_{#1}}

\newcommand{\alp}[2]{%
\ifthenelse{\equal{#1}{}}{\ifthenelse{\equal{#2}{}}{\alpha}{\alpha^{(#2)}}}{\ifthenelse{\equal{#2}{}}{\alpha_{#1}}{\alpha^{(#2)}_{#1}}}}
\newcommand{\alpi}[1]{%
\ifthenelse{\equal{#1}{}}{\innov{\alpv{}}}{\innov{\alpha}_{#1}}}
\newcommand{\aln}[1]{%
\ifthenelse{\equal{#1}{}}{\bar{\alpha}}{{\bar{\alpha}}^{(#1)}}}
\newcommand{\alpv}[1]{%
\ifthenelse{\equal{#1}{}}{ \pmb{\alpha}}{{\pmb{\alpha}}_{#1}}}
\newcommand{\alh}[1]{%
\ifthenelse{\equal{#1}{}}{ {\mathbf{a}}}{{{a}}_{#1}}}
\newcommand{\pauliV}[1]{%
\ifthenelse{\equal{#1}{}}{ {\vec{\pmb{\sigma}}}}{{\pmb{\sigma}}_{#1}}}
\newcommand\bR[1]{\mathbf{J}_{#1}}
\newcommand\bRf{J}
\newcommand\bRR[1]{\mathbf{R}_{#1}}
\newcommand\bRRf{R}
\newcommand{\bb}[1]{%
\ifthenelse{\equal{#1}{}}{ {\vec{\pmb{\rho}}}}{{\pmb{\rho}}_{#1}}}
\newcommand\bbb[1]{\pmb{\lambda}_{#1}}
\newcommand{\bRl}[1]{%
\ifthenelse{\equal{#1}{}}{\mathbf{R_l}}{\mathbf{R_l}_{#1}}%
}
\newcommand{\bRr}[1]{%
\ifthenelse{\equal{#1}{}}{\mathbf{R_r}}{\mathbf{R_r}_{#1}}%
}

\newcommand{\bRin}[1]{%
\ifthenelse{\equal{#1}{}}{{\bRf(\innov{\alpv{}})}}{{\bRf(\innov{\alpv{}}_{#1})}}}
\newcommand{\bRRin}[1]{%
\ifthenelse{\equal{#1}{}}{\bRRf(\innov{\alpv{}})}{\bRRf(\btin{k}, \innov{\alpv{}}_{#1})}}
\newcommand{\bRrin}[1]{%
\ifthenelse{\equal{#1}{}}{\innov{\bRr{}}}{\innov{\bRr{}}_{#1}}}
\newcommand{\bRlin}[1]{%
\ifthenelse{\equal{#1}{}}{\innov{\bRl{}}}{\innov{\bRl{}}_{#1}}}
\newcommand{\pn}[1]{%
\ifthenelse{\equal{#1}{}}{e^{i\phi}}{e^{i\phi_{#1}}}}
\newcommand{\pnin}[1]{%
\ifthenelse{\equal{#1}{}}{e^{i\innov{\phi}}}{e^{i\innov{\phi}_{#1}}}}
\newcommand{\bt}[1]{%
\ifthenelse{\equal{#1}{}}{\phi}{\phi_{#1}}}
\newcommand{\btin}[1]{%
\ifthenelse{\equal{#1}{}}{\innov{\phi}}{\innov{\phi}_{#1}}}
\newcommand\bM[1]{\mathbf{M}_{#1}}
\newcommand\bMf{M}
\newcommand{\bMin}[1]{%
\ifthenelse{\equal{#1}{}}{{\bMf({\innov{\alpv{}}})}}{{\bMf(\innov{\alpv{}}_{#1})}}}
\newcommand\bS[1]{\mathbf{s}_{#1}}
\newcommand{\autocorr}[2]{%
\ifthenelse{\equal{#1}{}}{\mathcal{A}_{#2}}{\mathcal{A}_{#2}{#1}}%
}
\newcommand{\autocorrinv}[2]{%
\ifthenelse{\equal{#1}{}}{\mathcal{A}^{-1}_{#2}}{\mathcal{A}^{-1}_{#2}{#1}}%
}
\newcommand{\thet}[2]{%
\ifthenelse{\equal{#1}{}}{\ifthenelse{\equal{#2}{}}{\theta}{\theta^{(#2)}}}{\ifthenelse{\equal{#2}{}}{\theta_{#1}}{\theta^{(#2)}_{#1}}}}
\newcommand\cpo[1]{\mathcal{K}({#1})}
\makeatletter
\newcommand{\sbullet}{%
  \hbox{\fontfamily{lmr}\fontsize{.4\dimexpr(\f@size pt)}{0}\selectfont\textbullet}}

\makeatother

\newcommand\innov[1]{\accentset{\sbullet}{#1}}
\newcommand{\rmT}{\mathrm{T}} 
\newcommand{\rmH}{\mathrm{H}} 
\newcommand{\rmC}{\mathrm{*}} 

\newcommand{\df}{\Delta \nu}

\newcounter{defcounter}
\newenvironment{myequation}{%
\addtocounter{equation}{-1}
\refstepcounter{defcounter}

\begin{equation*}}
{\end{equation*}}
\renewcommand{\thesection}{\Roman{section}}

\title{Polarization Drift Channel Model for Coherent Fibre-Optic Systems}

\author[1,*]{Cristian~B.~Czegledi}
\author[2]{Magnus~Karlsson}
\author[1]{Erik~Agrell}
\author[2]{Pontus~Johannisson}
\affil[1]{Chalmers University of Technology, Department of Signals and Systems, SE-41296 Gothenburg, Sweden}
\affil[2]{Chalmers University of Technology, Department of Microtechnology and Nanoscience, SE-41296 Gothenburg, Sweden}

\affil[*]{czegledi@chalmers.se}

\begin{abstract}
A theoretical framework is introduced to model the dynamical changes of the state of polarization during transmission in coherent fibre-optic  systems. The model generalizes the one-dimensional phase noise random walk to higher dimensions, accounting for random polarization drifts, emulating a random walk on the Poincar\'e sphere, which has been successfully verified using experimental data. The model is described in the Jones, Stokes and real four-dimensional formalisms, and the mapping between them is derived. Such a model will be increasingly important in simulating and optimizing future   systems,  where polarization-multiplexed transmission and sophisticated digital signal processing will be natural parts.  The proposed polarization drift model is the first of its kind as prior work either models polarization drift as a deterministic process or focuses on polarization-mode dispersion in systems where the state of polarization does not affect the receiver performance. We expect the model to be  useful  in a wide-range of photonics applications where stochastic polarization fluctuation is an issue.
\end{abstract}
\begin{document}

\flushbottom
\maketitle

\thispagestyle{empty}

\section*{Introduction}
{Enabled} by digital signal processing, coherent detection {allows} spectrally efficient communication based on quadrature amplitude modulation formats, which carry information in both the intensity and {the} phase of the optical field, in both polarizations. Polarization-multiplexed  quadrature phase-shift keying has  recently {been} introduced for 100~Gb/s transmission per channel and it is expected that in  the near future, higher-order modulation formats will become a necessity to reach {even} higher data rates.  {However, the demultiplexing of the  polarization-multiplexed channels in the receiver requires knowledge of the state of polarization (SOP), which is drifting with time. This implies that the SOP must be tracked
 by a dynamic equalizer \cite{Kim2009,Savory2010}. As the SOP changes with time in a random fashion, it is qualitatively different from the chromatic dispersion, which can be compensated for using a static equalizer. A deterministic or static behaviour would be straightforward to resolve, but with a nondeterministic SOP drift, the equalizer must be continuously adjusted.}

In order to fully understand the impact of an impairment on the performance of a transmission system, a channel model is {required}, which should {describe} the behaviour of the channel as accurately as possible. Based on the statistical information that such models reveal,  insights into how to treat the impairments  optimally in order to maximize performance can be obtained and used as a result. On the other hand, a channel model that does not {describe} the fibre accurately  may hinder the achievement of optimal performance. Therefore, it is important that the channel  model matches the stochastic nature of the fibre closely.

Very few results on modelling of \emph{random} polarization drifts are present in the literature. {In the context of equalizer development}, several  models have indeed been suggested, but typically by either using a \emph{constant} randomly chosen {SOP}  \cite{Kikuchi2008,Liu2009,Roudas2010,Johannisson2011a} or by generating cyclic/quasi-cyclic \emph{deterministic} changes \cite{Heismann1995,Savory2008,Muga2014,Louchet2014}, which are usually nonuniform in their coverage of the  possible SOPs.
{There is, on the other hand, a wealth of statistical models that describe differential group delay (DGD) and polarization mode dispersion  dating back to the eighties\cite{Poole1986}. However, these results were typically developed to be applicable in systems using intensity modulation or single-polarization (differential) phase-shift keying formats, which are not affected by the SOP drift. Thus, instead of modelling the time evolution of the SOP, differential equations describing the SOP change with frequency and fibre length were typically given \cite{Foschini1991, Brodsky2006}. There are also some direct measurements of SOP changes, e.g., by Soliman et al. \cite{Soliman2010} However, in their measurement, fast SOP changes were induced in a dispersion-compensating module under laboratory conditions, without considering the SOP drift of an entire fibre link. It can be concluded that stochastic SOP drift has so far not been given much attention.}

This paper  suggests a model for random SOP drifts in the time domain by generalizing the  one-dimensional (1D) phase noise random walk to a higher dimension. The model is based on a succession of random Jones matrices, where each matrix is parameterized by three random  variables, chosen from a zero-mean Gaussian distribution with a variance set by a \textit{polarization linewidth} parameter. The latter determines the speed of the drift and depends on the system details. The polarization drift has a random walk behaviour, where each step is independent of the previous steps and equally likely in all directions. The model is given in the Jones, Stokes formalisms and in the more general real 4D formalism\cite{Betti1991,Karlsson2014}.

{As argued above, the suggested model serves a different purpose than the models of polarization mode dispersion existing in the literature. In the latter case, the fibre is viewed either as a concatenation of a small number of segments with relatively high DGD, leading to the \emph{hinge model}\cite{Kogelnik2005, Brodsky2006, Antonelli2006a}, or the limiting case with segments of infinitesimal lengths, leading to a Maxwellian distribution of the DGD \cite{Curti1990}. The most common assumption is then to model the hinges as independent random SOP rotators \cite{Brodsky2006, Antonelli2006a}. In the \emph{anisotropic hinge model}, the SOP variation at the hinges is modelled as generalized waveplates parameterized by one random angle per hinge \cite{Li2008OL}. The \emph{hybrid hinge model} is a further generalized way to model the SOP changes at the hinges \cite{Schuster2014}, but in none of these publications is the SOP time evolution discussed.}

The {suggested} model can be used in simulations  for a wide range of photonics applications, where stochastic polarization fluctuation is an issue that needs to be modelled. For example, a fibre-optic communication link can be simulated, independently of the modulated data as well as of other considered impairments, which can be useful to, e.g.,  characterize receivers' performance. Moreover, it can reveal  statistical knowledge of the received samples affected by polarization rotations, based on which the  existing tracking algorithms   can be  optimally tuned  or new, more powerful  algorithms can be designed. High-precision transfer and remote synchronization of microwave and/or radio frequency signals \cite{Jung2014} is another application that could benefit from a better understanding of how it is affected by polarization drifts, which is currently the limiting factor towards a better performance. Other applications such as fibre-optic sensors \cite{Song2007} have been developed for use in a broad range of applications, fibre-optic gyroscopes \cite{Wang2013} and quantum key distribution \cite{Dynes2012}  are   strongly affected by polarization fluctuations and may benefit from a better understanding of the transmission medium.

\section*{Present Phase and Polarization Drift Models}
The fibre propagation through a linear medium is often described by a complex  $2 \times 2$ Jones matrix, which, neglecting the nonlinear phenomena, relates the received optical field to the input. Another approach  is to use the Stokes--Mueller formalism, where the evolution of the Stokes vectors is modelled by a Mueller  $3 \times 3$ unitary matrix. The latter has the advantage that the Stokes vectors are experimentally observable quantities and can be easily visualized as points on a  3D sphere, called the  \emph{Poincar\'e sphere}. A further approach to describe the SOP rotations exists in the 4D Euclidean space \cite{Betti1991}, where the wave propagation can be described  by a $4\times 4$ real unitary matrix \cite{Karlsson2014}.  We will focus most of our discussion on the Jones formalism and connect it to the Stokes and real 4D formalisms later on.  

An optical signal has two quadratures in two polarizations and can be described by a Jones vector as a function of  time $t$ and the propagation distance  $z$ as  $ \mathbf{E}(z,t) = (E_{x}(z,t) ,E_{y}(z,t))^\rmT$, where each element combines the real and imaginary parts of the electrical field in the $x$ and $y$ field components and $(\cdot)^{\rmT}$ denotes transposition. The $k$th discrete-time input into a transmission medium can be written as
$ \bx{k} = \mathbf{E}(0,kT)$,
where $T$  is the sampling time and it is commonly related to the symbol interval. The received discrete signal, at distance $L$, is 
$\by{k} = \mathbf{E}(L,kT)$.

The propagation of the optical field can be described by a  $2\times2$ complex-valued Jones matrix $\bR{k}$, which relates the received optical  field $\by{k}\in\mathbb{C}^2$, in the presence of optical amplifier noise and laser phase noise,  to the input $\bx{k}\in\mathbb{C}^2$ as
\begin{equation}\label{eq:sys_mod}
	\by{k}= \pn{k}\bR{k} \bx{k}+\bn{k},
\end{equation}
where $i=\sqrt{-1}$,  $\bt{k}$ models the carrier phase noise and $\bn{k}\in\mathbb{C}^2$ denotes the additive noise, which is represented by two independent complex circular zero-mean Gaussian random variables.  Assuming that  {polarization-dependent losses  are negligible,   the channel  matrix  $\bR{k}$ can be modelled as a unitary  matrix, which preserves the input power  during propagation. In this work, we assume that the chromatic dispersion has been compensated for and  polarization mode dispersion   is  negligible. The  transformation $\bR{k}$ can be described using} the matrix function
$\bRf(\alpv{k}{})$, which is  expressed using the \emph{matrix exponential} \cite[p.~165]{Bellman1960} parameterized by three variables $\alpv{}=(\alp{1}{}, \alp{2}{},\alp{3}{})$ according to {\cite{Frigo1986,Gordon2000}}  
\begin{equation}\label{eq_rot}
\begin{split}
\bRf(\alpv{}{}) &= \exp(-i\alpv{}\cdot\pauliV{}) \\
                 & = \mathbf{I}_2 \cos(\thet{}{}) - i \alh{}\cdot\pauliV{}\sin(\thet{}{}),
\end{split}
\end{equation}
where   $\pauliV{}=(\pauliV{1},\pauliV{2},\pauliV{3})$ is a tensor of the Pauli spin matrices \cite{Gordon2000}
{
\begin{equation}
\pauliV{1} = 
\begin{pmatrix}
  1 & 0 \\
  0&-1\\
\end{pmatrix};
\qquad
\pauliV{2} = 
\begin{pmatrix}
  0 & 1 \\
 1&0\\
\end{pmatrix};
\qquad
\pauliV{3} = 
\begin{pmatrix}
  0 & -i \\
  i &  0\\
\end{pmatrix}.
\end{equation}
This notation of $\pauliV{i}$  complies with the definition of the Stokes vector, and it is different from  the notation  introduced by Frigo\cite{Frigo1986}.} The operation $\alpv{}\cdot\pauliV{}$  should be interpreted as a scaled linear combination of the three matrices  $\pauliV{1},\pauliV{2},\pauliV{3}$, and $\mathbf{I}_2$ is the $2\times2$ identity matrix.
{In general, a $2\times2$ complex unitary matrix has four  degrees of freedom (DOFs), but in this case we explicitly factored out the phase noise. Including the identity matrix in $\pauliV{}$ would make it possible to account for all four DOFs.}
The vector $\alpv{}$ can be expressed  as a product of two independent random variables  $\alpv{}=\thet{}{}\alh{}$, i.e., its length $\thet{}{} = \norm*{\alpv{}}$ in the interval $[0,  \pi)$ and the  unit vector $\alh{} = (\alh{1}{}, \alh{2}{},\alh{3}{})$, which represents its direction on the unit sphere. 
Since $\bR{k}$ is unitary, the inverse can be found by the conjugate transpose operation or by negating  $\alpv{k}$, since
$
  \bR{k}^{-1} = \bR{k}^\rmH=  \bRf(-\alpv{k}),
$
. The same principle holds for the phase noise
$
  {(\pn{k})}^{-1}={(\pn{k})}^\rmH=e^{-i\bt{k}}.
$

{Modern coherent systems require a local oscillator at the receiver in order to get access to both phase and amplitude of the electrical field. The local oscillator serves as a reference but it is not synchronized with the transmitter laser, resulting in phase noise, which creates the need for  carrier synchronization.  The phase noise} is modelled by the angle $\bt{}$, while  $\alp{1}{},\alp{2}{} \text{ and } \alp{3}{}$ model random fluctuations of the SOP caused by fibre birefringence and coupling.
Both phase noise and SOP drift are dynamical processes that change randomly over time. The  SOP drift time can vary from microseconds up to seconds, depending on the link type. It is usually much longer than the phase drift time, which is in the microsecond range for modern coherent systems \cite{Simon1977}.

The update rule of the phase noise follows a Wiener process \cite{Pfau2009,Tur1985,Ip2007}
\begin{equation}\label{eq:ph_wiener}
  \bt{k}=\btin{k}+\bt{k-1},
\end{equation}
where $\btin{k}$ is the \textit{innovation} of the phase noise. An alternative {form} of equation (\ref{eq:ph_wiener}) is
\begin{equation}\label{eq:ph_update}
\pn{k}=\pnin{k}\pn{k-1},
\end{equation}
{which we will generalize later to account for the SOP drift.} The innovation $\btin{k}$ is a random variable drawn independently at each time instance $k$ from a zero-mean Gaussian distribution 
\begin{equation}\label{eq:ph_distr}
  \btin{k} \sim \mathcal{N}(0,\sigma_\nu^2),
\end{equation}
where  $\sigma_\nu^2 = 2\pi { \df} T$ and  $\df$ is the  sum of the linewidths of the transmitter and local oscillator lasers. 

The accumulated phase noise at time $k$ is the summation of $k$ Gaussian random variables $\btin{1},\ldots,\btin{k}$ and the initial phase $\bt{0}$, which becomes a Gaussian-distributed random variable with mean $\bt{0}$ and variance $ k \sigma_\nu^2$. The function $\pn{k}$ is periodic with period $2\pi$, which means that the phase angle $\bt{k}$ can be limited to the interval $[0, 2\pi)$ by applying the modulo $2\pi$ operation. In this case, the  probability density function (pdf) of $\bt{k}$ becomes a wrapped (around the unit circle) Gaussian distribution. 
It can be straightforwardly verified that the phase drift model has the following properties:
\begin{enumerate}
\item The innovation  $\btin{k}$ is \emph{independent} of $\btin{i}$ for $i=0, \dots, k-1$.
\item The innovation is \emph{symmetric}, in the sense that the probabilities for phase changes $\bt{}$ and $-\bt{}$  are equally likely.
\item The most likely next phase state is the current state.
\item The outcome of two consecutive steps $\pnin{1}\pnin{2}$ can be emulated by a single step $\pnin{\text{t}}$ by doubling the variance of $\btin{k}$.
\item As $k$ increases, the distribution of  $\pn{k}$ will approach the \emph{uniform distribution} on the unit circle. The convergence rate towards this distribution increases with the $\df T$ product. 
\end{enumerate}
The time evolution of the pdf of a fixed point corrupted by phase noise is exemplified in Fig.~\ref{fig:unif_phase}.

\begin{figure*}[t] 
\centering
\centering
\includegraphics[width=0.98\linewidth]{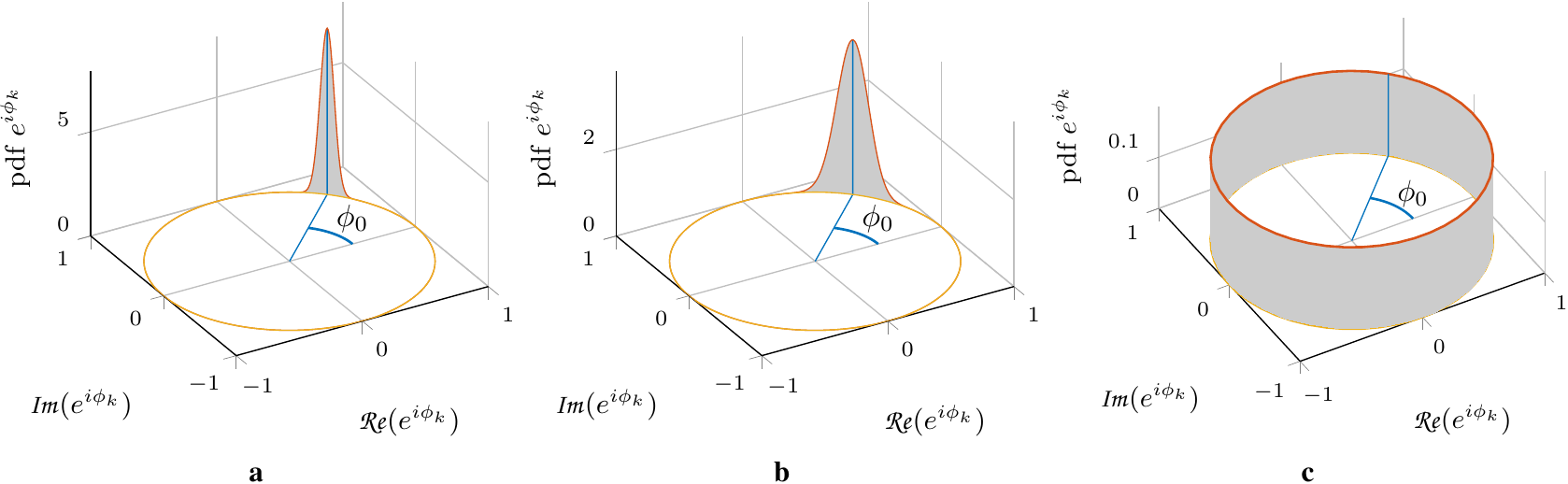}
\caption{Phase noise pdf evolution. The pdf of $\protect\pn{k}$ for $\protect\bt{0}=\pi/4$ and $\sigma_\nu^2=0.0025$ is shown. In ($\mathbf{a}$), $k=1$  corresponds to a single innovation and illustrates the second and third properties, i.e., the pdf is symmetric around the current state (the vertical line) and the peak of the pdf is at the current state.  In ($\mathbf{b}$),  $k=5$ and the pdf spreads over the circle. In ($\mathbf{c}$),  $k=8000$ and the pdf approaches the uniform pdf, which supports the last property. }
\label{fig:unif_phase}
\end{figure*}

The initial phase difference  between the two free running lasers has equal probability for every value, therefore it is common to model  $\bt{0}$  as a random variable uniformly distributed in the interval $[0, 2\pi)$, i.e., it has \emph{equal probability} for every possible state.

The autocorrelation function (ACF) quantifies the level of correlation between two samples of a random process taken at different time instances by taking the expected value $\mathbb{E}[\cdot]$ of the product of the samples. The ACF  of the phase noise with $lT$ time separation in response to a  constant input $\bx{}$ is \cite{Chorti2006}
\begin{equation}\label{eq:acf_pn}
  \begin{split}
  \autocorr{(l)}{\by{}} &= \mathbb{E}[\by{k}^\rmH\by{k+l}]\\
  &= \mathbb{E}[(\pn{k}\bx{})^\rmH\pn{k+l}\bx{}]\\
  &= \norm{\bx{}}^2\exp\big(-\frac{ \sigma_\nu^2\abs{l}}{2}\big)\\
  &= \norm{\bx{}}^2\exp\big(-{\pi \abs{l}T{ \df}}{}\big),
  \end{split}
\end{equation}
where the operations $\abs{\cdot}$ and  $\norm{\cdot}$ denote  absolute value and Euclidean norm, respectively. Here we neglected the SOP drift given by $\bR{k}$, in order to isolate the effects of the phase noise. We will compute the ACF of the SOP drift below.

In the literature, polarization rotations are generally modelled by the Jones matrix $\bR{k}$ in equation (\ref{eq:sys_mod}) or subsets of it, obtained by considering only one or two of the  three DOFs $\alp{1}{},\alp{2}{},\alp{3}{}$. Contrary to the phase noise, which is a random walk with respect to time,  the matrix $\bR{k}$ is in previous literature usually kept \emph{constant} \cite{Gong2011,Kikuchi2008,Liu2009,Roudas2010,Johannisson2011a}, or it follows a \emph{deterministic} cyclic/quasi-cyclic rotation pattern. The latter is obtained by modelling the parameters of $\bR{k}$ as frequency components $\omega k T$. For example,   $\alp{3}{}=\omega k T\,$ \cite{Johannisson2011,Savory2008,Muga2014} or $\alp{2}{},\alp{3}{}$ varied at different  frequencies \cite{Savory2010,Louchet2014,Heismann1995}.

\section*{Proposed Polarization  Drift Model}
\label{sec:proposed_model}
We propose to  model the polarization drift  as a sequence  of \emph{random} matrices, which exploit all \emph{three} DOFs of $\bR{k}$. The model simulates a random walk on the Poincar\'e sphere and we describe it in the Jones, Stokes and real 4D formalisms. 
In general, 4D unitary matrices  have \emph{six} DOFs, spanning a richer space than the  Jones (four DOFs) or Mueller  (three DOFs) matrices can. Out of the six DOFs, only four are physically realizable for propagating photons, and the remaining two are impossible to describe in the Jones or Stokes space \cite{Karlsson2014}. The Jones formalism can describe any physically possible phenomenon  in the optical fibre making it sufficient for wave propagation. The 4D representation is preferred in some digital communication scenarios since the performance of a constellation corrupted by additive noise can be directly quantified in this space in contrast to the Stokes formalism. Even though the extra two DOFs do not model lightwave propagation, they can be used to account for transmitter and/or receiver imperfections, which cannot be done using Jones or Mueller matrices. However, the extra two DOFs are out of the scope of this paper and we will focus on the remaining four. 
\subsection*{Jones Space Description}
\label{sec:jones-space-descr}

Similarly to the phase noise update equation (\ref{eq:ph_update}), we model the time evolution of $\bR{k}$ as
\begin{equation} \label{eq:time_dep}
  \bR{k} = \bRin{k}\bR{k-1},
\end{equation}
where $\bRin{k}$ is a random \emph{innovation} matrix defined as equation (\ref{eq_rot}). This mathematical formulation is not new to the polarization community, as it is commonly used to describe the polarization evolution in space (each Jones matrix represents a waveplate). However, here it models the temporal evolution of the SOP, where each innovation matrix corresponds to a time increment.
The parameters of the innovation $\bRin{k}$ are random and drawn independently from a zero-mean Gaussian distribution at each time instance $k$
\begin{equation}\label{eq:rand_alp}
  \alpi{}_k \sim \mathcal{N}(\mathbf{0},\sigma_p^2\, \mathbf{I}_3),
\end{equation}
where $\sigma_p^2= 2\pi { \Delta p}\, T$, and we refer to $\Delta p$ as the {\it polarization linewidth}, which quantifies the speed of the SOP drift, analogous to the linewidth describing the phase noise, cf. equation (\ref{eq:ph_distr}). 
Drawing $\alpv{}{}$ from a  zero-mean Gaussian distribution results in special cases of $\thet{}{} \text{ and }\alh{}$, where the former becomes a Maxwell--Boltzmann distributed random variable, and the vector $\alh{}$ is \emph{uniformly} distributed over the 3D unit sphere, implying that the marginal distribution of each $\alh{i}{}$ is uniform in $[-1,1]$ \cite{Muller1959}.

It is important to note that, contrary to phase noise, the equivalent vector $\alpv{k}{}$  parameterizing $\bR{k}$ in equation (\ref{eq:time_dep}) does not follow a Wiener process, i.e.,  $\alpv{k}{} \neq \alpi{}_k+\alpv{k-1}{} $, because in general $\bRf{}(\alpv{}_1)\bRf{}(\alpv{}_2)\neq \bRf{}(\alpv{}_1+\alpv{}_2)$. 
{Equality occurs for (anti-)parallel  $\alpv{}_1$  and  $\alpv{}_2$, and holds    approximately when  $\norm{\alpv{}_1}\ll 1$ and $\norm{\alpv{}_2}\ll 1$.} In a prestudy for this work \cite{Czegledi2015}, we incorrectly used a  Wiener process model for the polarization drift.  We will return to that model later and discuss its shortcomings.

\subsection*{Stokes Space Description}
\label{sec:stok-space-descr}

The evolution of the SOP is often analysed in the Stokes space, where the Jones vectors are expressed as real 3D Stokes vectors and can be easily visualized on the  Poincar\'e sphere. The transmitted Jones vector $\bx{k}$ can be expressed as a Stokes vector  \cite[eq.~(2.5.26)]{Damask2005} 
\begin{align} \label{eq:st_vec}
\bS{\bx{k}} &=\bx{k}^\rmH\pauliV{}\bx{k}\\
&=\begin{pmatrix}
\abs{E_{x}}^2-\abs{E_{y}}^2 \\
2\mathpzc{Re}(E_{x}E_{y}^\rmC) \\
-2\mathpzc{Im}(E_{x}E_{y}^\rmC) 
\end{pmatrix}, \nonumber
\end{align}
where the $i$th component of $\bS{\bx{k}}$ is given by $\bx{k}^\rmH\pauliV{i}\bx{k}$ and $(\cdot)^{\rmC}$ denotes conjugation. The equivalent Stokes propagation model of  equation (\ref{eq:sys_mod})  can be written  
\begin{equation}\label{eq:st_sys_mod}
  \bS{\by{k}} = \bM{k} \bS{\bx{k}} + \bS{\bn{k}}.
\end{equation}
It is important to note that only $\bS{\by{k}}$ and $\bS{\bx{k}}$ can be obtained by applying equation (\ref{eq:st_vec}) to $\by{k}$ and $\bx{k}$, respectively. The noise component $ \bS{\bn{k}}$  consists of three terms
\begin{equation}\label{eq:st_noise}
\bS{\bn{k}} = (\pn{k}\bR{k} \bx{k})^\rmH\pauliV{}\bn{k}+\bn{k}^\rmH\pauliV{}\pn{k}\bR{k} \bx{k} + \bn{k}^\rmH\pauliV{}\bn{k},
\end{equation}  
where the first two represent the signal--noise interaction and the last one the noise source.  It should be noted that there is no time averaging in equation (\ref{eq:st_vec}) and it represents an instantaneous  mapping to the Stokes space. Thus, the noise terms are polarized and rapidly varying.
 
The matrix  $\bM{k}$ is a $3\times3$ Mueller matrix, corresponding to the Jones  matrix $\bR{k}$, and the  polarization transformation introduced by it can be seen as a \emph{rotation} of the Poincar\'e sphere. It has three parameters, the same as $\bR{k}$, and can be written as 
$
  \bM{k} = \bMf(\alpv{k}), 
$
where the function $\bMf(\cdot)$ can be expressed using the matrix exponential {\cite{Gordon2000}}
\begin{equation} \label{eq:st_mat_exp}
\begin{split}
  \bMf{}(\alpv{}) &= \exp(2\cpo{\alpv{}}) \\
                  &= \exp(2\thet{}{}\cpo{\alh{}})\\
                  &= \mathbf{I}_3  + \sin(2\thet{}{}) \cpo{\alh{}} + (1- \cos(2\thet{}{})) \cpo{\alh{}}^2,
\end{split}
\end{equation}
where $\thet{}{} = \norm*{\alpv{}}$,  $\alh{}=\alpv{}/\thet{}{}$ and $\cpo{\alh{}}$ denotes 
\begin{equation}
  \cpo{\alh{}} =
\begin{pmatrix}
  0 &  -\alh{3}{} & \alh{2}{} \\
 \alh{3}{} &  0 & -\alh{1}{} \\
  -\alh{2}{} &  \alh{1}{} & 0 
\end{pmatrix}.
\end{equation}
The inverse can be obtained by 
$\bM{k}^{-1}=\bM{k}^{\rmT}=\bMf(-\alpv{k})$.
The transformation in equation (\ref{eq:st_mat_exp}) can be viewed in the axis-angle rotation description as a rotation around the unit vector $\alh{}$ by an angle $2\thet{}{}$.

Note that any operation that applies a phase rotation on both polarizations with the same amount, such as phase noise or frequency offsets, \emph{cannot} be modelled in the Stokes space. This can be seen in equation (\ref{eq:st_sys_mod}), where the phase noise does not alter the transmitted Stokes vector directly, but only contributes to the additive noise in equation (\ref{eq:st_noise}), which would not exist in the absence of $\bn{k}$.

Analogously to equation (\ref{eq:time_dep}), the time  evolution of $\bM{k}$ is
\begin{equation} \label{eq:st_time_dep}
  \bM{k} = \bMin{k}\bM{k-1},
\end{equation}
where $\bMin{k}$ is the innovation matrix parameterized by the random vector $\alpi{}_k$ defined in equation (\ref{eq:rand_alp}).

Fig.~\ref{fig:SOP_ex} shows an example of an SOP drift as it evolves with time. The line represents the evolution of the vector   $\bM{k}(0,0,1)^\rmT$ for $k=1,\ldots,3000$.

\begin{figure*}[t] 
\centering
\centering
\includegraphics[width=0.98\linewidth]{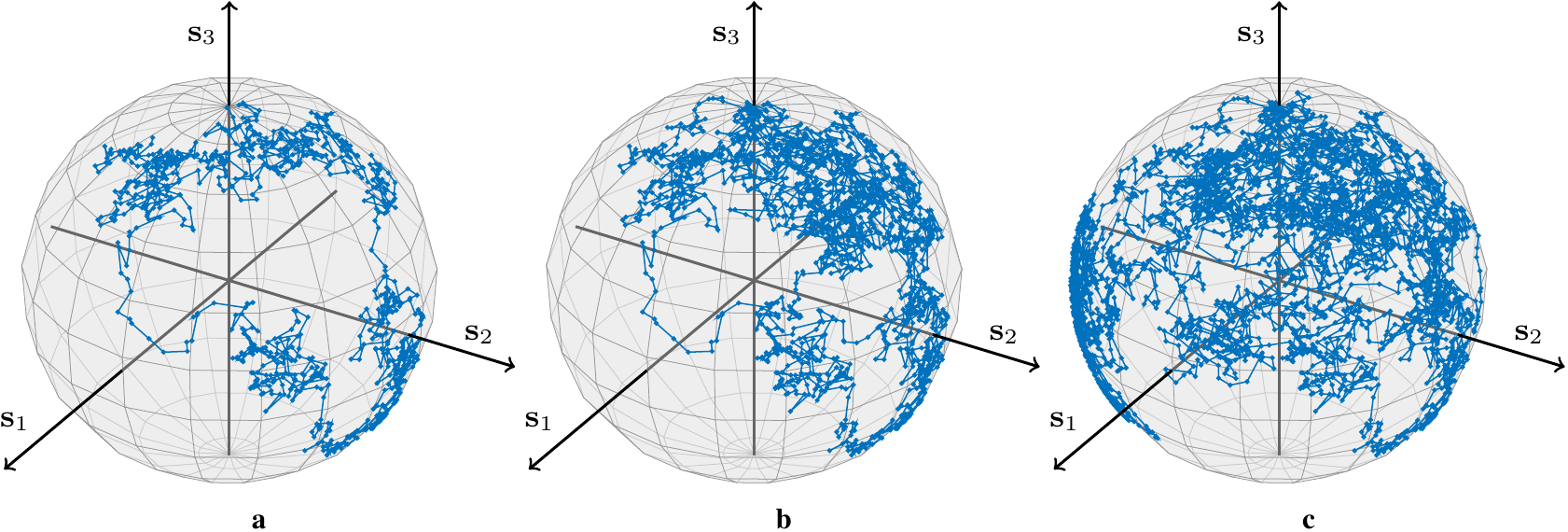}
\caption{Random walk. The evolution of a random SOP drift obtained by equation (\ref{eq:st_sys_mod}), without additive noise, for a fixed input $\bS{\bx{}}=(0,0,1)^\rmT$ and $\sigma_p^2=6\cdot 10^{-4}$ is shown.  The trajectories for ($\mathbf{a}$)  $k=1,\ldots, 300$, ($\mathbf{b}$) $k=1,\ldots, 1500$   and  ($\mathbf{c}$) $k=1,\ldots, 3000$ are plotted.}
\label{fig:SOP_ex}
\end{figure*}

\subsection*{4D Real Space Description}
\label{sec:4D-space-descr}
Another approach to express the time evolution of the phase and polarization drifts is to use the more uncommon  4D real formalism \cite{Karlsson2014,Betti1991,Agrell2009,Cusani1992,Karlsson2015}. {In this case, the transmitted/received sample and the noise term in equation (\ref{eq:sys_mod}) can all be represented as a real four-component vector   $( \mathpzc{Re}(E_{x}), \mathpzc{Im}(E_{x}), \mathpzc{Re}(E_{y}), \mathpzc{Im}(E_{y}))^\rmT$.}  The transformations induced by the channel are modelled by $4\times4$ real unitary matrices. 
The $\pn{k}\bR{k}$ transformation in equation (\ref{eq:sys_mod}) can be combined into 
\begin{equation} \label{eq:real_4D}
\bRR{k} = \bRRf(\bt{k},\alpv{k}),
\end{equation}
and the function $\bRRf(\cdot)$ can be described using the matrix exponential of a linear combination of four basis matrices \cite{Karlsson2014}
\begin{equation}\label{eq:4D_rot}
\begin{split}
\bRRf{}(\bt{},\alpv{}) &= \exp(-\bt{} \bbb{1}-\alpv{}\cdot\bb{}) \\
                      &= \exp(-\bt{} \bbb{1}) \exp(-\alpv{}\cdot\bb{})\\
                      & =(\mathbf{I}_4 \cos\bt{} - \bbb{1}\sin\bt{})(\mathbf{I}_4 \cos\thet{}{} - \alh{}\cdot{\bb{}}\sin\thet{}{}),
\end{split}
\end{equation}
where $\bb{}=(\bb{1},\bb{2},\bb{3})\text{ and }\bbb{1}$ are four constant basis matrices \cite[eqs. (20)--(23)]{Karlsson2014}{
\begin{equation}
\bb{1} = 
\begin{pmatrix}
  0 & -1 &  0 & 0\\
  1 &  0 &  0 & 0\\
  0 &  0 &  0 & 1\\
  0 &  0 & -1 & 0\\
\end{pmatrix};
\qquad
\bb{2} = 
\begin{pmatrix}
  0 &  0 &  0 & -1\\
  0 &  0 &  1 &  0\\
  0 & -1 &  0 &  0\\
  1 &  0 &  0 &  0\\
\end{pmatrix};
\qquad
\bb{3} = 
\begin{pmatrix}
  0 &  0 &  1 & 0\\
  0 &  0 &  0 & 1\\
 -1 &  0 &  0 & 0\\
  0 & -1 &  0 & 0\\
\end{pmatrix};
\qquad
\bbb{1}=
\begin{pmatrix}
  0 &  1 &  0 & 0\\
 -1 &  0 &  0 & 0\\
  0 &  0 &  0 & 1\\
  0 &  0 & -1 & 0\\
\end{pmatrix}.
\end{equation}
}
The inverse can be obtained as  
$\bRR{k}^{-1}=\bRR{k}^{\rmT}=\bRRf(-\bt{k},-\alpv{k})$.
The update of $\bRR{k}$ in equation (\ref{eq:real_4D}) can be expressed analogously to equations (\ref{eq:time_dep}) and  (\ref{eq:st_time_dep}) as
\begin{equation}\label{eq:4D_time_dep}
  \bRR{k} = \bRRin{k}\bRR{k-1},
\end{equation}
where the phase innovation $\btin{k}$ and the random vector $\alpi{}_k$ are defined by equation (\ref{eq:ph_distr}) and  (\ref{eq:rand_alp}), respectively. 

\begin{figure*}[t] 
\centering
\centering
\includegraphics[width=0.98\linewidth]{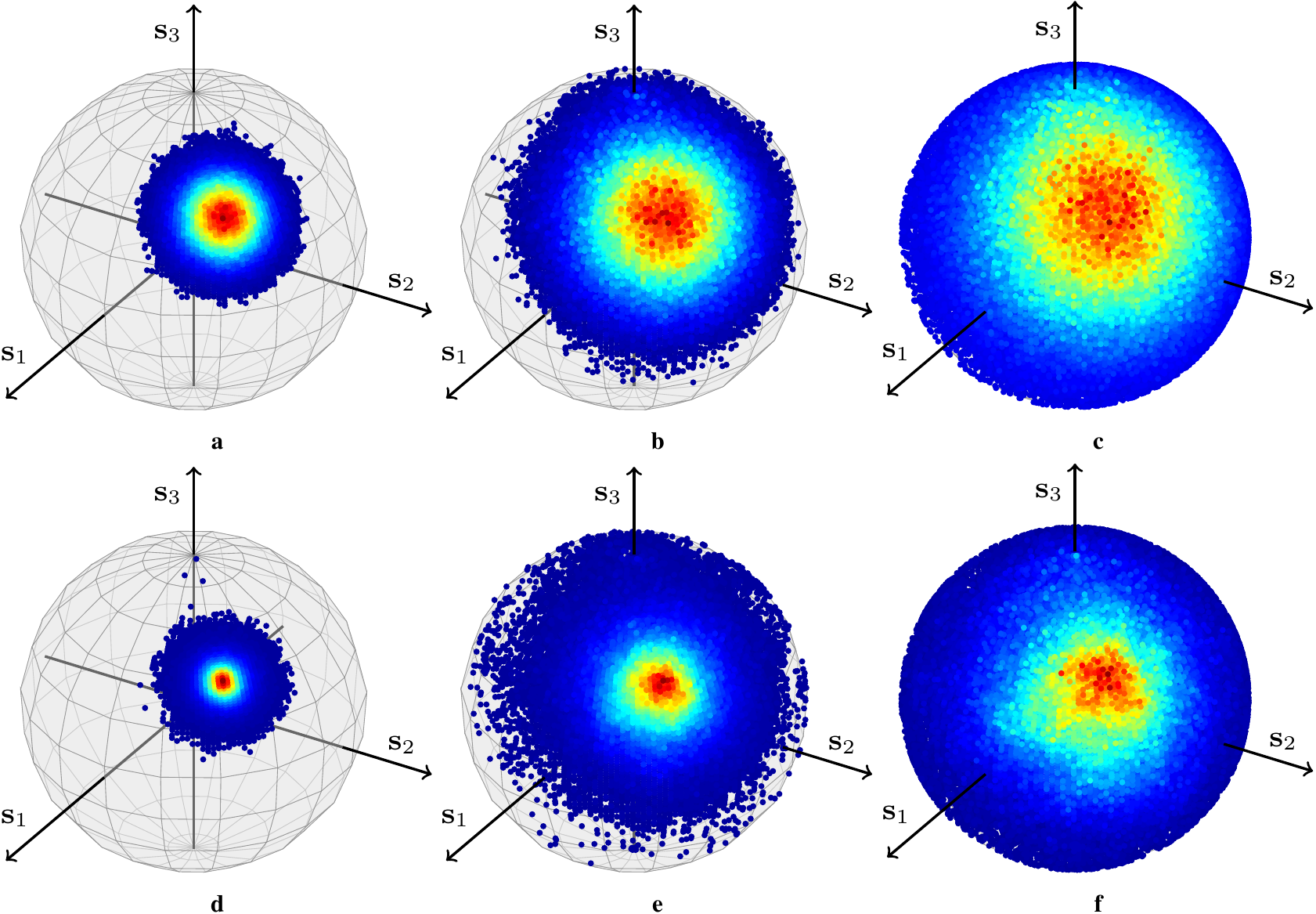}
\caption{The histograms of  $\protect \bM{k} \bS{\bx{}}$ for different steps $k$ and a fixed $\bS{\bx{}}=[1,1,1]^\rmT/\sqrt{3}$   obtained from the model (top row) and   from measurements (bottom row) are shown. The highest density is represented by dark red and the lowest by dark blue,  the outer part of the density. The parameters of the simulated drift in  equation (\ref{eq:rand_alp}) are $T=2.2$ h (set by the measurements) and $\Delta p=60$ nHz (obtained by fitting the  dash-dotted ACF line in  Fig. \ref{fig:autocorr}). In ($\mathbf{a}$) and ($\mathbf{d}$), $k=2$ innovation steps are plotted, whereas $k=8$ in  ($\mathbf{b}$), ($\mathbf{e}$) and $k=16$ in  ($\mathbf{c}$), ($\mathbf{f}$). Gaussian-like isotropic distributions can be noted in all cases, simulations and measurements, leading to a good (visual) agreement. The  spread over the sphere increases with $k$ and the pdf will become uniform if we let $k$ grow large enough. Unfortunately, our  measured data do not cover a long enough  time period  such that uniformity is achieved. }
\label{fig:SOP}
\end{figure*} 

\section*{Polarization Drift Model Properties}
\label{sec:model-validation}
Due to the similarities between the phase noise and the SOP drift, we will use the properties of the phase noise previously presented as guidelines to validate the proposed channel model. These properties can be reformulated as follows: 
\begin{enumerate}
\item The innovation $\alpi{}_k$ is \emph{independent} of $\alpi{}_i$ for $i=0, \dots, k-1$.
\item The innovation is \emph{isotropic}, in the sense that all possible orientations of the changes of the Stokes vector corresponding to a movement given by one innovation are equally likely.
\item The most likely next SOP is the current state.
\item The outcome of two consecutive steps $\bMin{1}\bMin{2}$ can be emulated by a single step $\bMin{\text{t}}$ by doubling the variance $\sigma_p^2$ of $\alpi{}_k$.
\item As $k$ increases, the pdf of the product $\bM{k} \bS{\bx{}}$, for a constant $\bS{\bx{}}$, will approach the \emph{uniform distribution} on the Poincar\'e sphere. The convergence rate towards this distribution increases with the ${ \Delta p} \, T$ product.
\end{enumerate}
 

In the following, we will investigate these properties  and discuss their validity.
\subsubsection*{Independent Innovations}
This can be easily concluded from the updating rule in equation (\ref{eq:time_dep}),  (\ref{eq:st_time_dep}) or  (\ref{eq:4D_time_dep}), where the parameters of the innovation do not depend on neither the previous innovations nor the state. 
\subsubsection*{Isotropic Innovation}
\label{sec:isotropic-innovation}
We will use the following theorem to show that the innovation $\bMin{}$ is isotropic.  
\newtheorem{theorem}{Theorem}
\begin{theorem}\label{th:1}
Let a random unit vector $\mathbf{a} \in \mathbb{R}^3$ be uniformly distributed over the 3D sphere, $\gamma$ be a random angle with an arbitrary pdf and  $\mathbf{x}\in \mathbb{R}^{3}$ an arbitrary unit vector.  The pdf of the vector $\mathbf{y}=\bMf{}(\gamma \mathbf{a})\mathbf{x}$, where $\bMf{}(\cdot)$ is defined in equation (\ref{eq:st_mat_exp}),  is invariant to rotations  around $\mathbf{x}$, i.e., $\bMf{}(\beta{}{} \mathbf{x})\mathbf{y}$ has the same pdf regardless of $\beta{}{}$. 
\end{theorem}
In other words, the theorem states that the pdf of the random vector $\mathbf{y}=\bMf{}(\gamma \mathbf{a})\mathbf{x}$ does not change if it is rotated by any angle around  $\mathbf{x}$. This can be true only if  $\mathbf{y}$ is isotropic (centred and equally likely in all directions) around  $\mathbf{x}$. The technical details of the proof are presented in Supplementary Section I. 

Note that Theorem \ref{th:1} holds for our proposed Stokes innovation matrix $\bMin{}$  since $\alpi{}=\thet{}{}\alh{}$ and  $\alh{}$ is uniformly distributed over  the 3D sphere, hence  the vector $\bMin{} \bS{\bx{}}$  is isotropic around $\bS{\bx{}}$. The evolution of $ \bM{k} \bS{\bx{}}$ can be seen as an \emph{isotropic random walk} on the Poincar\'e sphere starting at $ \bM{0} \bS{\bx{}}$ and taking random steps, of size proportional to $\sigma_p$, equally likely in all directions.  
Curiously, the isotropic property is given \emph{only} by the fact that the rotation axis $\alh{}$  is uniformly distributed over the sphere, while the angle $\thet{}{}$ may have \emph{any} pdf. In fact, the pdf  of the angle $\thet{}{}$ determines the shape of the pdf of $\bMin{} \bS{\bx{}}$, which we will discuss later. 
In contrast, our preliminary SOP drift model \cite{Czegledi2015} does not fulfil the isotropicity condition. The update method of the channel matrix was done by modelling  ${\alpv{k}{}}$ as Wiener processes, which does not fulfil Theorem \ref{th:1}.

\subsubsection*{Pdf of a Random Step}
Unfortunately, we were not able to find a closed form expression of the  pdf of the point $\bMin{} \bS{\bx{}}$ for a fixed $\bS{\bx{}}$ and a random $\bMin{}$. Using approximations, valid for  $\sigma_p^2\ll 1$, which is the case for most practical scenarios, the  pdf of  $\bMin{} \bS{\bx{}}$ can be approximated by a  bivariate Gaussian pdf centred at $\bS{\bx{}}$  of variance $\sigma_p^2\mathbf{I}_2$ on the plane normal to $\bS{\bx{}}$. The peak of the pdf is located at $\bS{\bx{}}$ and we can conclude that the next most likely SOP  is the current one. The details of the derivations are provided in Supplementary Section II.  In the same section, under the same assumption of small $\sigma_p^2$, we  show that the outcome of two consecutive random innovations can be achieved by a single random innovation if the variance of the  random variables $\alpi{}{}$ is doubled, which fulfils  property 4.

\subsubsection*{Uniformity}
The point  $ \bM{k} \bS{\bx{}}$  performs an isotropic random walk on the Poincar\'e sphere, therefore as the number of steps $k$ increases, the coverage of the  sphere will approach uniformity, meaning that all SOPs will be equally likely. This is  intuitive because taking random steps in all directions with no preferred orientation will lead to a uniform coverage. This property was observed by  Zhang et al.\cite{Zhang2007}, where  measurements of a  submarine cable resulted in uniform converge of the Poincar\'e sphere.

We compare the model with experimental results in Fig.~\ref{fig:SOP}, where  the evolution of the pdf of a Stokes vector affected by polarization drift after different numbers of innovations starting from a fixed point is shown. The experimental data was obtained by measuring the Jones matrices of a 127 km long buried fibre from 1505 nm to 1565 nm in steps of 0.1 nm for  36 days at every $\sim 2.2$ h. The technicalities of the measurement setup and postprocessing have been published elsewhere\cite{Karlsson2000}. The histograms corresponding to the measurements were captured from all the Stokes vectors obtained by applying equation (\ref{eq:st_vec}) to  the  product of the matrix $\bR{t}(\omega)  \bR{t-k}^{-1}(\omega)$, i.e.,  the measured  matrix corresponding to $k$ innovation steps, with a constant Jones vector, where $\bR{t}(\omega)$ denotes the measured Jones matrix at time  $t$ and wavelength $\omega$. 

\subsection*{Initial State}
\label{sec:initial-state}
The initial channel matrix $\bM{0}$ should be chosen such that all the SOPs on the Poincar\'e sphere are \emph{uniformly} distributed, i.e., equally likely, after the initial step $\bM{0}\bS{\bx{}}$, for any  $\bS{\bx{}}$. Analogously, the initial phase $\bt{0}$ in equation (\ref{eq:ph_wiener}) should be chosen from a uniform distribution in the interval $[0, 2\pi)$. 
In order to generate such a matrix $\bM{0}$, the axis $\alh{}$ must be uniformly distributed over the unit sphere and the distribution of the angle $\thet{}{}\in[0, \pi/2)$ must be {$(1- \cos\thet{}{})/\pi$} \cite{Rummler2002}. The generation of the angle  $\thet{}{}$  is not straightforward, and therefore we present an alternative to generate the axis and the angle simultaneously \cite{Karlsson2015}.   
The vector $\alpv{0}= \thet{}{}\alh{}$ is formed from  the unit vector $(\cos \thet{}{}, \alh{1}{}\sin \thet{}{}, \alh{2}{} \sin \thet{}{},\alh{3}{}\sin \thet{}{})^\rmT = \mathbf{g}/\norm{\mathbf{g}} $, where $\mathbf{g} \sim \mathcal{N}(\mathbf{0}, \mathbf{I}_4)$, which will satisfy the conditions for both axis $\alh{}{}$ and angle $\thet{}{}$. This approach of generating $\alpv{0}{}$ of the initial channel matrix can be used regardless of the considered method to  characterize the polarization evolution, i.e., Jones , Stokes or real 4D formalism.

\subsection*{The SOP Autocorrelation Functions}
The ACF is commonly used to quantify how quickly, on average, the channel changes  in time, frequency, etc. \cite{Vannucci2002,Soliman2013,Karlsson1999}. The ACF of the SOP drift  separated by a time difference of  $lT$ {in response to a constant input $\bx{}$} can be expressed as
\begin{equation}\label{eq:pol_ACF}
  \begin{split}
  \autocorr{(l)}{\by{}}  &= \mathbb{E}[\by{k}^\rmH\by{k+l}]\\
  &= \mathbb{E}[(\bR{k}\bx{})^\rmH\bR{k+l}\bx{}]\\
  & = \norm{\bx{}}^2\bigg((1-\sigma_p^2)\exp\Big(-\frac{\sigma_p^2}{2}\Big)\bigg)^{\abs{l}},
\end{split}
\end{equation}
where $\by{k}=\bR{k}\bx{}$ is the received sample.
The details of this derivation can be found in Supplementary Section III. Using the approximation $(1- \abs{l}\sigma_p^2/ \abs{l})^{ \abs{l}} \approx\exp (- \abs{l}\sigma_p^2)$, which is valid for $\sigma_p^2\ll 1$, equation (\ref{eq:pol_ACF}) can be approximated as
\begin{equation}\label{eq:pol_ACF_app}
  \autocorr{(l)}{\by{}}   \approx \norm{\bx{}}^2 \exp(-3 \pi \abs{l} T \Delta p).
\end{equation}

The ACF expressions in equations (\ref{eq:pol_ACF}) and (\ref{eq:pol_ACF_app}) hold for the Jones and real 4D space descriptions. Using an analogous derivation  as for equation (\ref{eq:pol_ACF}),  it can be shown that the ACF of the SOP drift in the Stokes space is
\begin{equation}\label{eq:pol_ACF_stokes}
  \begin{split}
  \autocorr{(l)}{\bS{}}  &= \mathbb{E}[  \bS{\by{k}}^\rmT  \bS{\by{k+l}}]\\
  &= \mathbb{E}[(\bM{k} \bS{\bx{}})^\rmT\bM{k+l} \bS{\bx{}}]\\
  & = \norm{ \bS{\bx{}}}^2  \bigg(\frac{2(1-4\sigma_p^2)\exp(-2\sigma_p^2)+1}{3}\bigg)^{\abs{l}},
\end{split}
\end{equation}
where $\bS{\by{k}}=\bM{k} \bS{\bx{}}$ is the Stokes received sample of a constant  input $\bS{\bx{}}$ affected by polarization drift. Using the following approximations: $\exp(-2\sigma_p^2)\approx 1-2\sigma_p^2$, $\sigma_p^4\approx 0$ and $(1-{8\pi  \abs{l} T \Delta p}/{ \abs{l}})^{ \abs{l}}\approx \exp(-8\pi  \abs{l} T \Delta p)$ for $\sigma_p^2\ll 1$, it can be approximated as 
\begin{equation}\label{eq:pol_ACF_stokes_app}
  \autocorr{(l)}{\bS{}}   \approx \norm{ \bS{\bx{}}}^2  \exp(-8 \pi \abs{l} T \Delta p).
\end{equation}
Both  ACFs  of the polarization drift depend only on the time separation $l$  but not on the absolute time $k$.
Comparing equation (\ref{eq:pol_ACF_app}) with  (\ref{eq:pol_ACF_stokes_app}), it is interesting to note that even if they describe the same physical  phenomenon, on average, a small movement in the Jones/real 4D space will result in a movement that is $\sqrt{8/3}$ larger  in the Stokes space. This  result was also previously observed in similar setups analysing  polarization-mode dispersion, where the autocorrelation, with respect to frequency,  was derived \cite{Vannucci2002,Karlsson2000}. 
The ACF of the phase noise equation (\ref{eq:acf_pn}) has the slowest decay rate, being a  factor of three slower than  equation (\ref{eq:pol_ACF_app}), and a factor of eight slower than equation (\ref{eq:pol_ACF_stokes_app}). 

Fig.~\ref{fig:autocorr} shows the ACFs of the phase noise and the SOP drift, and the later is compared with experimental results. Here the analytic equations (\ref{eq:pol_ACF}) and (\ref{eq:pol_ACF_stokes}) match the experiment for  $\Delta p=80$ nHz and $\Delta p=60$ nHz, respectively, and $T=2.2$ h. We attribute the discrepancy between the two values  to the nonunitary of the measured Jones matrices, which, through the nonlinear operation in equation (\ref{eq:st_vec}), cause the inconsistency. 

\begin{figure}[t]
\centering
\includegraphics[width=0.6\linewidth]{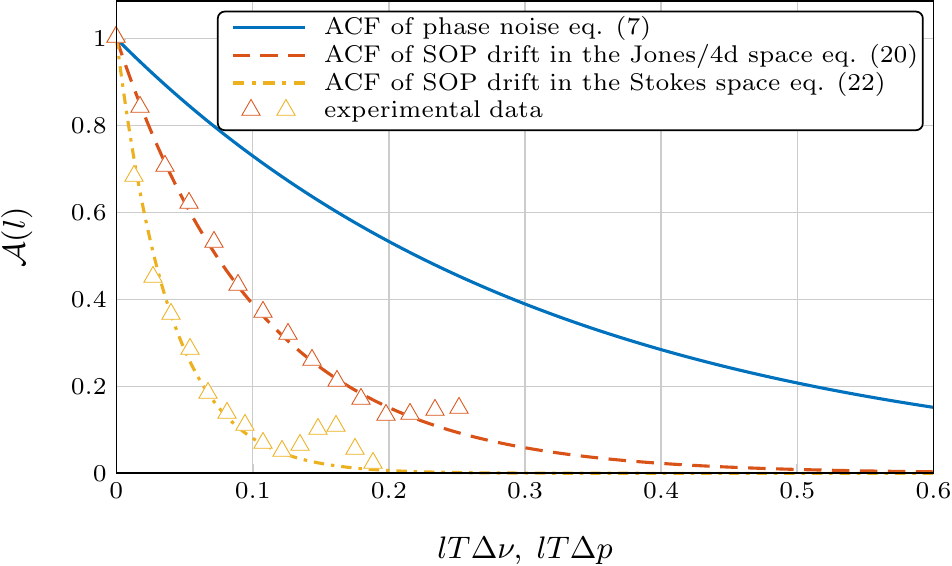}
\caption{ACF comparison. The normalized ACF of the phase noise and SOP drift is plotted versus normalized time. Solid/dashed lines refer to the analytic expressions, whereas the triangles are extracted from measurements. We  observe excellent  agreement between the experiment and theory, except in the tails of the experimental ACF. This inconsistency can be caused by the lack of accuracy in that region.  }
\label{fig:autocorr}
\end{figure}

\subsection*{Linewidth Parameters}
The choice of  $\df$ and  $\Delta p$ is important when a system is simulated considering phase noise and/or SOP drifts. Both parameters reflect physical properties of the system.  The quality of the (transmitter and receiver) lasers can be quantified by the $\df$ parameter, which represents the  sum of the linewidths of the deployed lasers. Distributed feedback lasers are commonly employed in transmission systems due to their cost efficiency. The linewidth of such lasers varies from 100 kHz to 10 MHz \cite{Pfau2009}.
The polarization linewidth parameter depends on the installation details. 
Measurements of polarization fluctuations  have been reported varying from   weeks (this work) to  seconds  \cite{Ogaki2003}, milliseconds \cite{Bulow1999,Krummrich2005} or microseconds under mechanical perturbations \cite{Krummirich2004}. 
The polarization linewidth  could be quantified through measurements of the ACF, either in the Jones or Stokes space, as in Fig.~\ref{fig:autocorr}. 

\section*{Summary}
This section is intended to summarize the previous sections and provide an easy implementable form of the proposed model without requiring knowledge about the details of the derivations.

First, the parameters of the channel must be selected:
\begin{itemize}
\item The sampling time $T$, often equal to the symbol interval.
\item The laser linewidth $\df$, where the contributions of the transmitter and receiver lasers are taken into account.
\item The polarization linewidth $\Delta p$.
\end{itemize}
The model is then implemented as follows:
\begin{enumerate}
\setcounter{enumi}{-1}
\item Set the initial state of the channel:
\begin{itemize}
  \item $\bt{0} \sim \mathcal{U}\{0, 2\pi\}$.
  \item $\bR{0} = \bR{}(\alpv{}_0) $ using equation (\ref{eq_rot}), where $\alpv{}_0 = \thet{}{}\alh{}$;  $ \thet{}{}\text{ and }\alh{}=(\alh{1},\alh{2},\alh{3})$ are identified from the unit vector \\ $(\cos \thet{}{}, \alh{1}{}\sin \thet{}{}, \alh{2}{} \sin \thet{}{},\alh{3}{}\sin \thet{}{})^\rmT = \mathbf{g}/\norm{\mathbf{g}} $, where $\mathbf{g} \sim \mathcal{N}(\mathbf{0}, \mathbf{I}_4)$. 
\end{itemize}
\item For $k\ge 1$, update the channel state:
  \begin{itemize}
    \item $\bt{k}=\btin{k}+\bt{k-1}$, where $\btin{k} \sim \mathcal{N}(0,2\pi\df T)$. 
    \item $ \bR{k} = \bRin{k}\bR{k-1}$, where $\bRin{k}$ is given by equation (\ref{eq_rot}) and $\alpi{}_k \sim \mathcal{N}(\mathbf{0},2\pi { \Delta p}\, T\, \mathbf{I}_3)$.
\end{itemize}
\item Apply the model for every transmitted sample $\bx{k}$, which results in the received sample $\by{k}$:
  \begin{itemize}
    \item $\by{k}= \pn{k}\bR{k} \bx{k}+\bn{k}$, where $\bn{k}$  denotes the additive noise represented by two independent complex circular zero-mean Gaussian random variables.
\end{itemize}
\end{enumerate}
Repeat the last two steps for every sample $\bx{k}$.

Alternatively, the Jones formalism used above can be replaced by the Stokes formalism using equation (\ref{eq:st_mat_exp}) instead of   equation (\ref{eq_rot}) and neglecting the phase noise $\pn{k}$; or  4D formalism by  combining  $\pn{k}$ and $\bR{k}$ into $\bRR{k}$ given by equation (\ref{eq:4D_rot}). In the former case, $\bx{k}$ and $\by{k}$ are 3D Stokes vectors defined as equation (\ref{eq:st_vec}), whereas in the latter case, $\bx{k}$, $\by{k}$ and $\bn{k}$ are 4D real vectors.

In the summary above, we have included phase  and additive noise for completeness. Without loss of generality, the model can be applied with/without phase and additive noise; also other impairments can be considered.

\section*{Conclusions}
\label{sec:concl}
We have proposed a channel model to emulate random polarization fluctuations based on a sequence of random matrices. The model is presented in the three formalisms, Jones, Stokes and real  4D, and it is   parameterized by a single  variable, called the polarization linewidth. 

The model has an isotropic behaviour, which has been successfully verified using experimental data, where every step on the Poincar\'e sphere is equally likely in all directions emulating an isotropic random walk and can be easily coupled with any other impairments to form a complete channel model.

Compared to the existing literature, the fundamental advantages of the proposed model are randomness and statistical uniformity.  Such a model is relevant for a wide range of fibre-optical applications where stochastic polarization fluctuations are an issue.  It can potentially lead to improved signal processing that accounts optimally for this impairment and more realistic simulations can be carried out in order to accurately quantify system performance.



\bibliographystyle{naturemag}

\begin{thebibliography}{10}
\expandafter\ifx\csname url\endcsname\relax
  \def\url#1{\texttt{#1}}\fi
\expandafter\ifx\csname urlprefix\endcsname\relax\def\urlprefix{URL }\fi
\providecommand{\bibinfo}[2]{#2}
\providecommand{\eprint}[2][]{\url{#2}}

\bibitem{Kim2009}
\bibinfo{author}{Kim, R.} \emph{et~al.}
\newblock \bibinfo{title}{{Performance of dual-polarization QPSK for optical
  transport systems}}.
\newblock \emph{\bibinfo{journal}{Journal of Lightwave Technology}}
  \textbf{\bibinfo{volume}{27}}, \bibinfo{pages}{3546--3559}
  (\bibinfo{year}{2009}).

\bibitem{Savory2010}
\bibinfo{author}{Savory, S.~J.}
\newblock \bibinfo{title}{{Digital coherent optical receivers: algorithms and
  subsystems}}.
\newblock \emph{\bibinfo{journal}{IEEE Journal of Selected Topics in Quantum
  Electronics}} \textbf{\bibinfo{volume}{16}}, \bibinfo{pages}{1164--1179}
  (\bibinfo{year}{2010}).

\bibitem{Kikuchi2008}
\bibinfo{author}{Kikuchi, K.}
\newblock \bibinfo{title}{{Polarization-demultiplexing algorithm in the digital
  coherent receiver}}.
\newblock In \emph{\bibinfo{booktitle}{Digest of the IEEE/LEOS Summer Topical
  Meetings}}, \bibinfo{pages}{MC22} (\bibinfo{address}{Acapulco, Mexico},
  \bibinfo{year}{2008}).

\bibitem{Liu2009}
\bibinfo{author}{Liu, L.} \emph{et~al.}
\newblock \bibinfo{title}{{Initial tap setup of constant modulus algorithm for
  polarization de-multiplexing in optical coherent receivers}}.
\newblock In \emph{\bibinfo{booktitle}{Conference on Optical Fiber
  Communication}}, \bibinfo{pages}{OMT2} (\bibinfo{address}{San Diego, CA},
  \bibinfo{year}{2009}).

\bibitem{Roudas2010}
\bibinfo{author}{Roudas, I.} \emph{et~al.}
\newblock \bibinfo{title}{{Optimal polarization demultiplexing for coherent
  optical communications systems}}.
\newblock \emph{\bibinfo{journal}{Journal of Lightwave Technology}}
  \textbf{\bibinfo{volume}{28}}, \bibinfo{pages}{1121--1134}
  (\bibinfo{year}{2010}).

\bibitem{Johannisson2011a}
\bibinfo{author}{Johannisson, P.} \emph{et~al.}
\newblock \bibinfo{title}{{Convergence comparison of the CMA and ICA for blind
  polarization demultiplexing}}.
\newblock \emph{\bibinfo{journal}{Journal of Optical Communications and
  Networking}} \textbf{\bibinfo{volume}{3}}, \bibinfo{pages}{493--501}
  (\bibinfo{year}{2011}).

\bibitem{Heismann1995}
\bibinfo{author}{Heismann, F.} \& \bibinfo{author}{Tokuda, K.~L.}
\newblock \bibinfo{title}{{Polarization-independent electro-optic
  depolarizer}}.
\newblock \emph{\bibinfo{journal}{Optics letters}}
  \textbf{\bibinfo{volume}{20}}, \bibinfo{pages}{1008--1010}
  (\bibinfo{year}{1995}).

\bibitem{Savory2008}
\bibinfo{author}{Savory, S.~J.}
\newblock \bibinfo{title}{{Digital filters for coherent optical receivers}}.
\newblock \emph{\bibinfo{journal}{Optics express}}
  \textbf{\bibinfo{volume}{16}}, \bibinfo{pages}{804--817}
  (\bibinfo{year}{2008}).

\bibitem{Muga2014}
\bibinfo{author}{Muga, N.~J.} \& \bibinfo{author}{Pinto, A.~N.}
\newblock \bibinfo{title}{{Adaptive 3-D Stokes space-based polarization
  demultiplexing algorithm}}.
\newblock \emph{\bibinfo{journal}{Journal of Lightwave Technology}}
  \textbf{\bibinfo{volume}{32}}, \bibinfo{pages}{3290--3298}
  (\bibinfo{year}{2014}).

\bibitem{Louchet2014}
\bibinfo{author}{Louchet, H.}, \bibinfo{author}{Kuzmin, K.} \&
  \bibinfo{author}{Richter, A.}
\newblock \bibinfo{title}{{Joint carrier-phase and polarization rotation
  recovery for arbitrary signal constellations}}.
\newblock \emph{\bibinfo{journal}{IEEE Photonics Technology Letters}}
  \textbf{\bibinfo{volume}{26}}, \bibinfo{pages}{922--924}
  (\bibinfo{year}{2014}).

\bibitem{Poole1986}
\bibinfo{author}{Poole, C.} \& \bibinfo{author}{Wagner, R.}
\newblock \bibinfo{title}{{Phenomenological approach to polarisation dispersion
  in long single-mode fibres}}.
\newblock \emph{\bibinfo{journal}{Electronics Letters}}
  \textbf{\bibinfo{volume}{22}}, \bibinfo{pages}{1029--1030}
  (\bibinfo{year}{1986}).

\bibitem{Foschini1991}
\bibinfo{author}{Foschini, G.~J.} \& \bibinfo{author}{Poole, C.~D.}
\newblock \bibinfo{title}{{Statistical theory of polarization dispersion in
  single mode fibers}}.
\newblock \emph{\bibinfo{journal}{Journal of Lightwave Technology}}
  \textbf{\bibinfo{volume}{9}}, \bibinfo{pages}{1439--1456}
  (\bibinfo{year}{1991}).

\bibitem{Brodsky2006}
\bibinfo{author}{Brodsky, M.}, \bibinfo{author}{Frigo, N.~J.},
  \bibinfo{author}{Boroditsky, M.} \& \bibinfo{author}{Tur, M.}
\newblock \bibinfo{title}{{Polarization mode dispersion of installed fibers}}.
\newblock \emph{\bibinfo{journal}{Journal of Lightwave Technology}}
  \textbf{\bibinfo{volume}{24}}, \bibinfo{pages}{4584--4599}
  (\bibinfo{year}{2006}).

\bibitem{Soliman2010}
\bibinfo{author}{Soliman, G.}, \bibinfo{author}{Reimer, M.} \&
  \bibinfo{author}{Yevick, D.}
\newblock \bibinfo{title}{{Measurement and simulation of polarization
  transients in dispersion compensation modules}}.
\newblock \emph{\bibinfo{journal}{Journal of the Optical Society of America A}}
  \textbf{\bibinfo{volume}{27}}, \bibinfo{pages}{2532--2541}
  (\bibinfo{year}{2010}).

\bibitem{Betti1991}
\bibinfo{author}{Betti, S.}, \bibinfo{author}{Curti, F.}, \bibinfo{author}{{De
  Marchis}, G.} \& \bibinfo{author}{Iannone, E.}
\newblock \bibinfo{title}{{A novel multilevel coherent optical system:
  4-quadrature signaling}}.
\newblock \emph{\bibinfo{journal}{Journal of Lightwave Technology}}
  \textbf{\bibinfo{volume}{9}}, \bibinfo{pages}{514--523}
  (\bibinfo{year}{1991}).

\bibitem{Karlsson2014}
\bibinfo{author}{Karlsson, M.}
\newblock \bibinfo{title}{{Four-dimensional rotations in coherent optical
  communications}}.
\newblock \emph{\bibinfo{journal}{Journal of Lightwave Technology}}
  \textbf{\bibinfo{volume}{32}}, \bibinfo{pages}{1246--1257}
  (\bibinfo{year}{2014}).

\bibitem{Kogelnik2005}
\bibinfo{author}{Kogelnik, H.} \emph{et~al.}
\newblock \bibinfo{title}{{First-order PMD outage for the hinge model}}.
\newblock \emph{\bibinfo{journal}{IEEE Photonics Technology Letters}}
  \textbf{\bibinfo{volume}{17}}, \bibinfo{pages}{1208--1210}
  (\bibinfo{year}{2005}).

\bibitem{Antonelli2006a}
\bibinfo{author}{Antonelli, C.} \& \bibinfo{author}{Mecozzi, A.}
\newblock \bibinfo{title}{{Theoretical characterization and system impact of
  the hinge model of PMD}}.
\newblock \emph{\bibinfo{journal}{Journal of Lightwave Technology}}
  \textbf{\bibinfo{volume}{24}}, \bibinfo{pages}{4064--4074}
  (\bibinfo{year}{2006}).

\bibitem{Curti1990}
\bibinfo{author}{Curti, F.}, \bibinfo{author}{Daino, B.},
  \bibinfo{author}{de~Marchis, G.} \& \bibinfo{author}{Matera, F.}
\newblock \bibinfo{title}{Statistical treatment of the evolution of the
  principal states of polarization in single-mode fibers}.
\newblock \emph{\bibinfo{journal}{Journal of Lightwave Technology}}
  \textbf{\bibinfo{volume}{8}}, \bibinfo{pages}{1162--1166}
  (\bibinfo{year}{1990}).

\bibitem{Li2008OL}
\bibinfo{author}{Li, J.}, \bibinfo{author}{Biondini, G.},
  \bibinfo{author}{Kath, W.~L.} \& \bibinfo{author}{Kogelnik, H.}
\newblock \bibinfo{title}{Anisotropic hinge model for polarization-mode
  dispersion in installed fibers}.
\newblock \emph{\bibinfo{journal}{Optics Letters}}
  \textbf{\bibinfo{volume}{33}}, \bibinfo{pages}{1924--1926}
  (\bibinfo{year}{2008}).

\bibitem{Schuster2014}
\bibinfo{author}{Schuster, J.}, \bibinfo{author}{Marzec, Z.},
  \bibinfo{author}{Kath, W.~L.} \& \bibinfo{author}{Biondini, G.}
\newblock \bibinfo{title}{Hybrid hinge model for polarization-mode dispersion
  in installed fiber transmission systems}.
\newblock \emph{\bibinfo{journal}{Journal of Lightwave Technology}}
  \textbf{\bibinfo{volume}{32}}, \bibinfo{pages}{1412--1419}
  (\bibinfo{year}{2014}).

\bibitem{Jung2014}
\bibinfo{author}{Jung, K.} \emph{et~al.}
\newblock \bibinfo{title}{{Frequency comb-based microwave transfer over fiber
  with $7\times 10^{-19}$ instability using fiber-loop optical-microwave phase
  detectors}}.
\newblock \emph{\bibinfo{journal}{Optics Letters}}
  \textbf{\bibinfo{volume}{39}}, \bibinfo{pages}{1577--1580}
  (\bibinfo{year}{2014}).

\bibitem{Song2007}
\bibinfo{author}{Song, H.~B.} \emph{et~al.}
\newblock \bibinfo{title}{{Polarization fluctuation suppression and sensitivity
  enhancement of an optical correlation sensing system}}.
\newblock \emph{\bibinfo{journal}{Measurement Science and Technology}}
  \textbf{\bibinfo{volume}{18}}, \bibinfo{pages}{3230--3234}
  (\bibinfo{year}{2007}).

\bibitem{Wang2013}
\bibinfo{author}{Wang, X.}, \bibinfo{author}{He, Z.} \&
  \bibinfo{author}{Hotate, K.}
\newblock \bibinfo{title}{{Automated suppression of polarization fluctuation in
  resonator fiber optic gyro with twin $90^\circ$ polarization-axis rotated
  splices}}.
\newblock \emph{\bibinfo{journal}{Journal of Lightwave Technology}}
  \textbf{\bibinfo{volume}{31}}, \bibinfo{pages}{366--374}
  (\bibinfo{year}{2013}).

\bibitem{Dynes2012}
\bibinfo{author}{Dynes, J.~F.} \emph{et~al.}
\newblock \bibinfo{title}{{Stability of high bit rate quantum key distribution
  on installed fiber}}.
\newblock \emph{\bibinfo{journal}{Optics Express}}
  \textbf{\bibinfo{volume}{20}}, \bibinfo{pages}{16339--16347}
  (\bibinfo{year}{2012}).

\bibitem{Bellman1960}
\bibinfo{author}{Bellman, R.}
\newblock \emph{\bibinfo{title}{Introduction to Matrix Analysis}}
  (\bibinfo{publisher}{McGraw-Hill}, \bibinfo{address}{New York, NY},
  \bibinfo{year}{1960}).

\bibitem{Frigo1986}
\bibinfo{author}{Frigo, N.~J.}
\newblock \bibinfo{title}{{A generalized geometrical representation of coupled
  mode theory}}.
\newblock \emph{\bibinfo{journal}{IEEE Journal of Quantum Electronics}}
  \textbf{\bibinfo{volume}{22}}, \bibinfo{pages}{2131--2140}
  (\bibinfo{year}{1986}).

\bibitem{Gordon2000}
\bibinfo{author}{Gordon, J.~P.} \& \bibinfo{author}{Kogelnik, H.}
\newblock \bibinfo{title}{{PMD fundamentals: Polarization mode dispersion in
  optical fibers}}.
\newblock \emph{\bibinfo{journal}{Proceedings of the National Academy of
  Sciences of the United States of America}} \textbf{\bibinfo{volume}{97}},
  \bibinfo{pages}{4541--4550} (\bibinfo{year}{2000}).

\bibitem{Simon1977}
\bibinfo{author}{Simon, A.} \& \bibinfo{author}{Ulrich, R.}
\newblock \bibinfo{title}{{Evolution of polarization along a single-mode
  fiber}}.
\newblock \emph{\bibinfo{journal}{Applied Physics Letters}}
  \textbf{\bibinfo{volume}{31}}, \bibinfo{pages}{517--520}
  (\bibinfo{year}{1977}).

\bibitem{Pfau2009}
\bibinfo{author}{Pfau, T.}, \bibinfo{author}{Hoffmann, S.} \&
  \bibinfo{author}{Noe, R.}
\newblock \bibinfo{title}{{Hardware-efficient coherent digital receiver concept
  with feedforward carrier recovery for \textit{M}-QAM constellations}}.
\newblock \emph{\bibinfo{journal}{Journal of Lightwave Technology}}
  \textbf{\bibinfo{volume}{27}}, \bibinfo{pages}{989--999}
  (\bibinfo{year}{2009}).

\bibitem{Tur1985}
\bibinfo{author}{Tur, M.}, \bibinfo{author}{Moslehi, B.} \&
  \bibinfo{author}{Goodman, J.}
\newblock \bibinfo{title}{{Theory of laser phase noise in recirculating
  fiber-optic delay lines}}.
\newblock \emph{\bibinfo{journal}{Journal of Lightwave Technology}}
  \textbf{\bibinfo{volume}{3}}, \bibinfo{pages}{20--31} (\bibinfo{year}{1985}).

\bibitem{Ip2007}
\bibinfo{author}{Ip, E.} \& \bibinfo{author}{Kahn, J.~M.}
\newblock \bibinfo{title}{{Feedforward carrier recovery for coherent optical
  communications}}.
\newblock \emph{\bibinfo{journal}{Journal of Lightwave Technology}}
  \textbf{\bibinfo{volume}{25}}, \bibinfo{pages}{2675--2692}
  (\bibinfo{year}{2007}).

\bibitem{Chorti2006}
\bibinfo{author}{Chorti, A.} \& \bibinfo{author}{Brookes, M.}
\newblock \bibinfo{title}{{A spectral model for RF oscillators with power-law
  phase noise}}.
\newblock \emph{\bibinfo{journal}{IEEE Transactions on Circuits and Systems I:
  Regular Papers}} \textbf{\bibinfo{volume}{53}}, \bibinfo{pages}{1989--1999}
  (\bibinfo{year}{2006}).

\bibitem{Gong2011}
\bibinfo{author}{Gong, C.}, \bibinfo{author}{Wang, X.},
  \bibinfo{author}{Cvijetic, N.} \& \bibinfo{author}{Yue, G.}
\newblock \bibinfo{title}{{A novel blind polarization demultiplexing algorithm
  based on correlation analysis}}.
\newblock \emph{\bibinfo{journal}{Journal of Lightwave Technology}}
  \textbf{\bibinfo{volume}{29}}, \bibinfo{pages}{1258--1264}
  (\bibinfo{year}{2011}).

\bibitem{Johannisson2011}
\bibinfo{author}{Johannisson, P.} \emph{et~al.}
\newblock \bibinfo{title}{{Modified constant modulus algorithm for
  polarization-switched QPSK}}.
\newblock \emph{\bibinfo{journal}{Optics Express}}
  \textbf{\bibinfo{volume}{19}}, \bibinfo{pages}{7734--41}
  (\bibinfo{year}{2011}).

\bibitem{Muller1959}
\bibinfo{author}{Muller, M.~E.}
\newblock \bibinfo{title}{{A note on a method for generating points uniformly
  on n-dimensional spheres}}.
\newblock \emph{\bibinfo{journal}{Communications of the Associations for
  Computing Machinery}} \textbf{\bibinfo{volume}{2}}, \bibinfo{pages}{19--20}
  (\bibinfo{year}{1959}).

\bibitem{Czegledi2015}
\bibinfo{author}{Czegledi, C.~B.}, \bibinfo{author}{Agrell, E.} \&
  \bibinfo{author}{Karlsson, M.}
\newblock \bibinfo{title}{{Symbol-by-symbol joint polarization and phase
  tracking in coherent receivers}}.
\newblock In \emph{\bibinfo{booktitle}{Optical Fiber Communication
  Conference}}, \bibinfo{pages}{W1e.3} (\bibinfo{address}{Los Angeles, CA},
  \bibinfo{year}{2015}).

\bibitem{Damask2005}
\bibinfo{author}{Damask, J.~N.}
\newblock \emph{\bibinfo{title}{{Polarization Optics in Telecommunications}}}
  (\bibinfo{publisher}{Springer}, \bibinfo{address}{New York, NY},
  \bibinfo{year}{2005}).

\bibitem{Agrell2009}
\bibinfo{author}{Agrell, E.} \& \bibinfo{author}{Karlsson, M.}
\newblock \bibinfo{title}{{Power-efficient modulation formats in coherent
  transmission systems}}.
\newblock \emph{\bibinfo{journal}{Journal of Lightwave Technology}}
  \textbf{\bibinfo{volume}{27}}, \bibinfo{pages}{5115--5126}
  (\bibinfo{year}{2009}).

\bibitem{Cusani1992}
\bibinfo{author}{Cusani, R.}, \bibinfo{author}{Iannone, E.},
  \bibinfo{author}{Salonico, A.~M.} \& \bibinfo{author}{Todaro, M.}
\newblock \bibinfo{title}{{An efficient multilevel coherent optical system:
  M-4Q-QAM}}.
\newblock \emph{\bibinfo{journal}{Journal of Lightwave Technology}}
  \textbf{\bibinfo{volume}{10}}, \bibinfo{pages}{777--786}
  (\bibinfo{year}{1992}).

\bibitem{Karlsson2015}
\bibinfo{author}{Karlsson, M.}, \bibinfo{author}{Czegledi, C.~B.} \&
  \bibinfo{author}{Agrell, E.}
\newblock \bibinfo{title}{{Coherent transmission channels as 4d rotations}}.
\newblock In \emph{\bibinfo{booktitle}{Signal Processing in Photonics
  Communications}}, \bibinfo{pages}{SpM3E.2} (\bibinfo{address}{Boston, MA},
  \bibinfo{year}{2015}).

\bibitem{Zhang2007}
\bibinfo{author}{Zhang, Z.}, \bibinfo{author}{Bao, X.}, \bibinfo{author}{Yu,
  Q.} \& \bibinfo{author}{Chen, L.}
\newblock \bibinfo{title}{{Time evolution of PMD due to tides and sun radiation
  on submarine fibers}}.
\newblock \emph{\bibinfo{journal}{Optical Fiber Technology}}
  \textbf{\bibinfo{volume}{13}}, \bibinfo{pages}{62--66}
  (\bibinfo{year}{2007}).

\bibitem{Karlsson2000}
\bibinfo{author}{Karlsson, M.}, \bibinfo{author}{Brentel, J.} \&
  \bibinfo{author}{Andrekson, P.~A.}
\newblock \bibinfo{title}{{Long-term measurement of PMD and polarization drift
  in installed fibers}}.
\newblock \emph{\bibinfo{journal}{Journal of Lightwave Technology}}
  \textbf{\bibinfo{volume}{18}}, \bibinfo{pages}{941--951}
  (\bibinfo{year}{2000}).

\bibitem{Rummler2002}
\bibinfo{author}{Rummler, H.}
\newblock \bibinfo{title}{{On the distribution of rotation angles how great is
  the mean rotation angle of a random rotation?}}
\newblock \emph{\bibinfo{journal}{The Mathematical Intelligencer}}
  \textbf{\bibinfo{volume}{24}}, \bibinfo{pages}{6--11} (\bibinfo{year}{2002}).

\bibitem{Vannucci2002}
\bibinfo{author}{Vannucci, A.} \& \bibinfo{author}{Bononi, A.}
\newblock \bibinfo{title}{{Statistical characterization of the Jones matrix of
  long fibers affected by polarization mode dispersion (PMD)}}.
\newblock \emph{\bibinfo{journal}{Journal of Lightwave Technology}}
  \textbf{\bibinfo{volume}{20}}, \bibinfo{pages}{811--821}
  (\bibinfo{year}{2002}).

\bibitem{Soliman2013}
\bibinfo{author}{Soliman, G.} \& \bibinfo{author}{Yevick, D.}
\newblock \bibinfo{title}{{Temporal autocorrelation functions of PMD variables
  in the anisotropic hinge model}}.
\newblock \emph{\bibinfo{journal}{Journal of Lightwave Technology}}
  \textbf{\bibinfo{volume}{31}}, \bibinfo{pages}{2976--2980}
  (\bibinfo{year}{2013}).

\bibitem{Karlsson1999}
\bibinfo{author}{Karlsson, M.} \& \bibinfo{author}{Brentel, J.}
\newblock \bibinfo{title}{{Autocorrelation function of the polarization-mode
  dispersion vector}}.
\newblock \emph{\bibinfo{journal}{Optics Letters}}
  \textbf{\bibinfo{volume}{24}}, \bibinfo{pages}{939--941}
  (\bibinfo{year}{1999}).

\bibitem{Ogaki2003}
\bibinfo{author}{Ogaki, K.}, \bibinfo{author}{Nakada, M.},
  \bibinfo{author}{Nagao, Y.} \& \bibinfo{author}{Nishijima, K.}
\newblock \bibinfo{title}{{Fluctuation differences in the principal states of
  polarization in aerial and buried cables}}.
\newblock In \emph{\bibinfo{booktitle}{Optical Fiber Communications
  Conference}}, \bibinfo{pages}{MF13} (\bibinfo{address}{Atlanta, GA},
  \bibinfo{year}{2003}).

\bibitem{Bulow1999}
\bibinfo{author}{B\"ulow, H.} \emph{et~al.}
\newblock \bibinfo{title}{{Measurement of the maximum speed of PMD fluctuation
  in installed field fiber}}.
\newblock In \emph{\bibinfo{booktitle}{Optical Fiber Communication Conference
  and the International Conference on Integrated Optics and Optical Fiber
  Communication}}, \bibinfo{pages}{WE4.1/83} (\bibinfo{address}{San Diego, CA},
  \bibinfo{year}{1999}).

\bibitem{Krummrich2005}
\bibinfo{author}{Krummrich, P.}, \bibinfo{author}{Schmidt, E.-D.},
  \bibinfo{author}{Weiershausen, W.} \& \bibinfo{author}{Mattheus, A.}
\newblock \bibinfo{title}{{Field trial results on statistics of fast
  polarization changes in long haul WDM transmission systems}}.
\newblock In \emph{\bibinfo{booktitle}{Optical Fiber Communication
  Conference}}, \bibinfo{pages}{OThT6} (\bibinfo{address}{Anaheim, CA},
  \bibinfo{year}{2005}).

\bibitem{Krummirich2004}
\bibinfo{author}{Krummirich, P.} \& \bibinfo{author}{Kotten, K.}
\newblock \bibinfo{title}{{Extremely fast (microsecond timescale) polarization
  changes in high speed long haul WDM transmission systems}}.
\newblock In \emph{\bibinfo{booktitle}{Optical Fiber Communication
  Conference}}, \bibinfo{pages}{FI3} (\bibinfo{address}{Los Angeles, CA},
  \bibinfo{year}{2004}).

\end{thebibliography}

\begin{thebibliography}{1}

\bibitem{Lapidoth2003}
A.~Lapidoth and S.~M. Moser, ``{Capacity bounds via duality with applications
  to multiple-antenna systems on flat-fading channels},'' {\em IEEE
  Transactions on Information Theory}, vol.~49, pp.~2426--2467, Oct. 2003.

\bibitem{Shynk2013}
J.~J. Shynk, {\em Probability, Random Variables, and Random Processes: Theory
  and Signal Processing Applications}.
\newblock Hoboken, NJ: John Wiley \& Sons, 2013.

\end{thebibliography}

\section*{Acknowledgements}
C. B. Czegledi would like to thank R. Devassy for suggesting the derivation in Supplementary Section I and inspiring discussions. This research was supported by the Swedish Research Council (VR) under grants no. 2010-4236 and 2012-5280, and performed within the Fiber Optic Communications Research Center (FORCE) at Chalmers.

\newpage
\vspace{0.4cm}
{\centering
{ \LARGE Supplementary Information}
\vspace{0.0cm}

\begin{center}{
\large Cristian~B.~Czegledi, Magnus~Karlsson, Erik~Agrell \\  and~Pontus~Johannisson
}
\end{center}
}

\section{Isotropicity} \label{app:isotropic}
We will first prove a lemma, which will be used to prove Theorem 1 in the main paper.
\newtheorem{lem}{Lemma}
\begin{lem}\label{lem:1}
For all unit vectors  $\mathbf{a}_1, \mathbf{a}_2 \in \mathbb{R}^{3}$ and angles $\gamma_1, \gamma_2\in \mathbb{R}$ we have
\begin{myequation} \label{eq:lemm1}
\bMf{}(\gamma_1 \mathbf{a}_1)  \bMf{}(\gamma_2 \mathbf{a}_2)  = \bMf{}(\gamma_2 \bMf{}(\gamma_1 \mathbf{a}_1)  \mathbf{a}_2) \bMf{}(\gamma_1 \mathbf{a}_1),  \tag{S1} 
\end{myequation}
where $\bMf{}(\cdot)$ is defined as equation (13).
\end{lem}

\begin{proof}
To prove it, we will use the fact that the cross product $\cpo{\mathbf{a}} \mathbf{b}$, where $\cpo{\mathbf{\cdot}}$ is defined in equation~(14),  is invariant under rotations $ \bM{}$
\begin{myequation}
  \bM{} \cpo{\mathbf{a}} \mathbf{b} =      \cpo{\bM{}\mathbf{a}} \bM{}\mathbf{b},\tag{S2}
\end{myequation}for all  $\mathbf{a},\mathbf{b} \in \mathbb{R}^{3}$ and any unitary matrix $\bM{}$ defined as equation (13). This is true since the vector $\cpo{\mathbf{a}} \mathbf{b}$ is orthogonal to $ \mathbf{a}$ and $ \mathbf{b}$ and its length depends on the area given by the parallelogram formed by  $ \mathbf{a}$ and $ \mathbf{b}$. These properties are invariant under rotations, hence so is the cross product. Consequently,
\begin{myequation}\label{eq:he}
  \bM{} \cpo{\mathbf{a}} =      \cpo{\bM{}\mathbf{a}} \bM{}. \tag{S3}
\end{myequation}

Now let us denote $\mathbf{v} = \bMf{}(\gamma_1 \mathbf{a}_1)\mathbf{a}_2$ and, based on equation (\ref{eq:he}), we can write
\begin{align*}
      \cpo{\mathbf{v}} \bMf{}(\gamma_1 \mathbf{a}_1) & =  \cpo{\bMf{}(\gamma_1 \mathbf{a}_1)\mathbf{a}_2} \bMf{}(\gamma_1 \mathbf{a}_1) \\ 
                                  & = \bMf{}(\gamma_1 \mathbf{a}_1) \cpo{\mathbf{a}_2}.  \tag{S4} \label{eq:S8}
\end{align*}

By applying equation (\ref{eq:S8}) twice, we obtain
\begin{align*}
      \cpo{\mathbf{v}}^2 \bMf{}(\gamma_1 \mathbf{a}_1)  & =  \cpo{\mathbf{v}}\bMf{}(\gamma_1 \mathbf{a}_1) \cpo{\mathbf{a}_2} \\
               &= \bMf{}(\gamma_1 \mathbf{a}_1) \cpo{\mathbf{a}_2}^2.  \tag{S5} \label{eq:s9}
\end{align*}

The right-hand side of equation (\ref{eq:lemm1}) can be simplified as 
\begin{align*}
    \bMf{}(\gamma_2  \bMf{}(\gamma_1 \mathbf{a}_1)  \mathbf{a}_2) \bMf{}(\gamma_1 \mathbf{a}_1)  
    & = \bMf{}(\gamma_2 \mathbf{v}) \bMf{}(\gamma_1 \mathbf{a}_1)   \\
    & = ( \mathbf{I}_3  + \sin(2\gamma_2) \cpo{\mathbf{v}} +  (1- \cos(2\gamma_2)) \cpo{\mathbf{v}}^2) \bMf{}(\gamma_1 \mathbf{a}_1) \tag{S6a} \label{eq:S9.3} \\    
    & =\bMf{}(\gamma_1 \mathbf{a}_1)  + \sin(2\gamma_2) \cpo{\mathbf{v}} \bMf{}(\gamma_1 \mathbf{a}_1)+  (1- \cos(2\gamma_2)) \cpo{\mathbf{v}}^2 \bMf{}(\gamma_1 \mathbf{a}_1)\\    
    & =\bMf{}(\gamma_1 \mathbf{a}_1)  + \sin(2\gamma_2) \bMf{}(\gamma_1 \mathbf{a}_1) \cpo{\mathbf{a}_2} +  (1- \cos(2\gamma_2)) \bMf{}(\gamma_1 \mathbf{a}_1)\cpo{\mathbf{a}_2}^2 \tag{S6b} \label{eq:s9.5} \\    
    & =\bMf{}(\gamma_1 \mathbf{a}_1) (\mathbf{I}_3   + \sin(2\gamma_2)\cpo{\mathbf{a}_2} +  (1- \cos(2\gamma_2)) \cpo{\mathbf{a}_2}^2 )  \\ 
    & =\bMf{}(\gamma_1 \mathbf{a}_1) \bMf{}(\gamma_2 \mathbf{a}_2), \tag{S6c} \label{eq:s10}
\end{align*}
where  equations (\ref{eq:S9.3}) and  (\ref{eq:s10}) follow from equation (13) and in equation (\ref{eq:s9.5}) we used equations (\ref{eq:S8}) and  (\ref{eq:s9}).
\end{proof}

Now we have the necessary tools to prove  Theorem 1 in the main paper.
\begin{proof}
The theorem can be proved by showing that $\bMf{}(\beta{}{} \mathbf{x})\mathbf{y}\sim \mathbf{y}$ for any real angle $\beta{}{}$. 
Let $\mathbf{z}=\bMf{}(\beta{}{} \mathbf{x})\mathbf{y}$, which can be expressed as
\begin{align*}
    \mathbf{z} &=\bMf{}(\beta{}{} \mathbf{x})\bMf{}(\gamma \mathbf{a})\mathbf{x} \\
               &=\bMf{}(\gamma \bMf{}(\beta{}{} \mathbf{x})\mathbf{a})\bMf{}(\beta{}{} \mathbf{x})\mathbf{x} \tag{S7a} \label{eq:s10.5}\\
               &=\bMf{}(\gamma \bMf{}(\beta{}{} \mathbf{x})\mathbf{a})\mathbf{x},  \tag{S7b}
\end{align*}
where in equation (\ref{eq:s10.5}) we used Lemma \ref{lem:1}. Since the vector $\mathbf{a}$ is uniformly distributed over the 3D sphere, the vector $\bMf{}(\beta{}{} \mathbf{x})\mathbf{a}$ is also uniformly distributed over the 3D sphere \cite[Def. 6.18]{Lapidoth2003}, which makes  $\mathbf{z}\sim \mathbf{y}$.
\end{proof}


\section{3D Distribution Analysis} \label{app:small_sig}
In this section, we derive the approximate pdf of the point $\bS{\by{}}=\bMf{}(\alpi{}) \bS{\bx{}}$ for a fixed $\bS{\bx{}}$ and random $\bMf{}(\alpi{})$. Exact expressions are very difficult to obtain, therefore we make use of the  approximations $\sin 2\thet{}{}\approx 2\thet{}{} $ and $\cos 2\thet{}{} \approx 1$, valid for $\sigma_p^2\ll 1$, i.e., $\norm{\alpi{}{}}\ll 1$. Thus $\bMf{}(\alpv{})$ in equation (13) can be approximated as 
\begin{align*}
   \bMf{}(\alpv{})  & \approx \mathbf{I}_3  + 2\thet{}{} \cpo{\alh{}} \\
& \approx
\begin{pmatrix}
 1 &  -2\thet{}{} \alh{3}{} & 2\thet{}{} \alh{2}{} \\
 2\thet{}{} \alh{3}{} &  1  & -2\thet{}{} \alh{1}{} \\
 -2\thet{}{} \alh{2}{} &   2\thet{}{} \alh{1}{} & 1
\end{pmatrix} \\
& \approx
\begin{pmatrix}
 1 &  -2\alp{3}{} & 2\alp{2}{} \\
 2\alp{3}{} &  1  & -2\alp{1}{} \\
 -2\alp{2}{} &   2\alp{1}{} & 1
\end{pmatrix}. \tag{S8} \label{eq:s12}
\end{align*}


Without loss of generality,   we simplify the analysis by setting $\bS{\bx{}}=(1,0,0)^{\mathrm{T}}$. In this case, based on equation (\ref{eq:s12}),   $\bS{\by{}}=\bMf{}(\alpi{}) \bS{\bx{}}=(1, 2\alpi{3}{} , -2\alpi{2}{} )^{\mathrm{T}}$ and it can be noted that $\bS{\by{}}$ then has a bivariate Gaussian distribution on the plane normal to $\bS{\bx{}}$ and the peak of the distribution centred at $\bS{\bx{}}$. 

Using equation (\ref{eq:s12}) and by removing high order terms, such as $\alp{i}{}\alp{j}{}$ for any $i,j$, the multiplication of two matrices $\bMf{}(\alpv{})$  can be approximated as
\begin{align*}
   \bMf{}(\alpv{}) \bMf{}(\pmb{\beta})&\approx
\begin{pmatrix}
 1 &  -2\alp{3}{}-2\beta_3 & 2\alp{2}{} +2\beta_2\\
 2\alp{3}{}+2\beta_3 &  1  & -2\alp{1}{}-2\beta_1 \\
 -2\alp{2}{} -2\beta_2&   2\alp{1}{} +2\beta_1& 1
\end{pmatrix}\\
&\approx \bMf{}(\alpv{}+\pmb{\beta}). \tag{S9} \label{eq:s13}
\end{align*} 
From equation (\ref{eq:s13}) we can conclude that two consecutive small innovations  can be replaced by a single innovation, i.e., $\bMin{1}\bMin{2}\bS{\bx{}}\sim\bMin{\text{t}}\bS{\bx{}}$, by doubling the variance  $\sigma_p^2$.

\section{Autocorrelation} \label{app:autocorr}
In this section, we derive the ACF of the SOP drift in equation (20). The derivation uses the Jones description of the model but the result  is valid for the 4D description as well.

At first we will calculate the expectation of the innovation matrix from equation (2)
\begin{align*}
     \mathbb{E}[\bRin{}] &= \mathbb{E}[ \mathbf{I}_2 \cos(\thet{}{}) -i(\alh{1}{}\pauliV{1}+\alh{2}{}\pauliV{2}+\alh{3}{}\pauliV{3})\sin(\thet{}{})]\\
                        &= \mathbb{E}[\cos(\thet{}{})] \mathbf{I}_2   \tag{S10a} \label{eq:1}\\
                        &= \int^\infty_0 \cos(\thet{}{}) f_{\thet{}{}}(\thet{}{}) \mathrm{d}\thet{}{} \mathbf{I}_2 \\
                        & = \bigg((1-\sigma_p^2)\exp\Big(-\frac{\sigma_p^2}{2}\Big)\bigg) \mathbf{I}_2, \tag{S10b} \label{eq:2}
\end{align*}
where $\alpi{} = (\alpi{1}, \alpi{2}, \alpi{3}) \sim \mathcal{N}(\mathbf{0},\sigma_p^2\, \mathbf{I}_3)$, $\thet{}{} = \norm*{\alpi{}}$ and  $\alh{}=\alpi{}/\thet{}{}=(\alh{1},\alh{2},\alh{3})$. Equation (\ref{eq:1}) follows because $\mathbb{E}[\alh{i}]=0$ and the random variables $\alh{i}$, $\thet{}{}$ are independent. Equation (\ref{eq:2}) follows because the  probability density function (pdf) of $\thet{}{}$ is \cite[eq. (3.195)]{Shynk2013}
 \begin{myequation}
  f_{\thet{}{}}(\thet{}{})=\frac{1}{\sigma_p^3} \sqrt{\frac{2}{\pi}}\thet{}{}^2\exp\Big(-\frac{\thet{}{}^2}{2\sigma_p^2}\Big), \tag{S11}
\end{myequation}
for $\thet{}{}\ge 0$.

The ACF of $\by{k}$ at time separation $l\ge 0$ for a constant input $\bx{}$ is
\begin{align*}
\autocorr{(k, k+l)}{\by{}} &= \mathbb{E}[\by{k}^{\rmH}\by{k+l}] \\
            &= \bx{}^{\rmH} \mathbb{E}[\bR{k}^{\rmH} \bR{k+l}]\bx{} \\
            &= \bx{}^{\rmH} \mathbb{E}[\bR{k}^{\rmH} \bRin{k+l}\dots \bRin{k+1} \bR{k}]\bx{} \\
            &= \bx{}^{\rmH} \mathbb{E}[\bRin{}]^l\mathbb{E}[\bR{k}^{\rmH} \bR{k}]\bx{} \tag{S12a} \label{eq:s2.5}\\
            &= \bx{}^{\rmH} \mathbb{E}[\bRin{}]^l\bx{}. \tag{S12b} \label{eq:s3}
\end{align*}
In equation (\ref{eq:s2.5}) we used the fact that the expectation of the innovation matrix is a scaled identity matrix (equation (\ref{eq:2})) that  commutes with $\bR{k}^{\rmH}$, and the fact that the innovation matrices $\bRin{k}$ are independent.

Using equation (\ref{eq:2}) in equation (\ref{eq:s3}), the ACF can be expressed as
\begin{align*}
  \autocorr{(l)}{\by{}} & = \bx{}^{\rmH} \Bigg(\bigg((1-\sigma_p^2)\exp\Big(-\frac{\sigma_p^2}{2}\Big)\bigg) \mathbf{I}_2\Bigg)^{{l}}\bx{} \\
        & = \norm{\bx{}}^2\bigg((1-\sigma_p^2)\exp\Big(-\frac{\sigma_p^2}{2}\Big)\bigg)^{\abs{l}}. \tag{S13} \label{eq:s4}
\end{align*}
For symmetry reasons, the absolute value of $\abs{l}$ replaced $l$ in equation (\ref{eq:s4}), making the expression valid for negative $l$ as well.

\end{document}